\documentclass[10pt]{article}
\usepackage{setup}

\usepackage[subtle]{savetrees}
\usepackage[left=1in, right=1in, top=1in, bottom=1in]{geometry}
\usepackage[utf8]{inputenc}
\usepackage{authblk}
\usepackage{subfig}
\usepackage[normalem]{ulem}

\usepackage{amsthm}
\usepackage{amssymb}
\usepackage{amsmath}
\usepackage{mathtools}
\usepackage{csquotes}
\usepackage{enumitem}   
\usepackage{tikz}
\usepackage{graphicx}
\usepackage{xspace}
\usepackage{xcolor}
\usepackage{thmtools}
\usepackage{thm-restate}

\usetikzlibrary{arrows.meta}
\usepackage[colorlinks=true,allcolors=blue,allbordercolors=white]{hyperref} % hyperlink setup

\usepackage{caption}  % caption formatting package
\setlist[enumerate]{itemsep=\smallskipamount,parsep=0pt,label={\rm \roman*)}}
\setlist[itemize]{itemsep=\smallskipamount,parsep=0pt}
\renewcommand{\defn}[1]{{\textit{\textbf{\boldmath #1}}}\xspace}
\renewcommand{\paragraph}[1]{\vspace{0.09in}\noindent{\bf \boldmath #1.}}

\usepackage{bbm}

\newcommand{\interior}[1]{ {\kern0pt#1}^{\mathrm{o}} }

\usepackage{algorithm}
\usepackage[noend]{algpseudocode} % adding [noend] deletes the end while and stuff

\usepackage[capitalise,nameinlink,noabbrev]{cleveref}
\crefname{equation}{}{} % cref{eq:blah} only does (1) instead of Equation (1)
\crefname{enumi}{Step}{} % cref{eq:blah} only does (1) instead of Item(1)

\theoremstyle{plain}
\newtheorem{theorem}{Theorem}[section]

\newtheorem{lemma}[theorem]{Lemma}
\newtheorem{claim}[theorem]{Claim}
\newtheorem{corollary}[theorem]{Corollary}

\theoremstyle{definition}
\newtheorem{fact}[theorem]{Fact}
\newtheorem{definition}[theorem]{Definition}

\theoremstyle{remark}
\newtheorem{remark}[theorem]{Remark}

\usepackage{mdframed}
\usepackage{framed}

\def\draft{1}
\newcommand{\anote}[1]{\ifnum\draft=1\textcolor{teal}{[\textbf{Anna:} #1]}\fi}
\newcommand{\nnote}[1]{\ifnum\draft=1\textcolor{purple}{[\textbf{Nathan:} #1]}\fi}
\newcommand{\bnote}[1]{\ifnum\draft=1\textcolor{orange}{[\textbf{Byron:} #1]}\fi}
\newcommand{\dnote}[1]{\ifnum\draft=1\textcolor{pink}{[\textbf{Divya:} #1]}\fi}

\newcommand{\dumb}{\textsf{GREEDY}\xspace}
\newcommand{\degreedy}{\textsf{DEGREEDY}\xspace}
\newcommand{\shortsighted}{\textsf{SHORTSIGHTED}\xspace}
\newcommand{\bruteforce}{\textsf{BRUTE-FORCE}\xspace}

\title{Matching Algorithms in the Sparse Stochastic Block Model}
\author{Anna Brandenberger, Byron Chin, Nathan S. Sheffield and Divya Shyamal}

\affiliation{Department of Mathematics, Massachusetts Institute of Technology}
\date{August 2023}
\begin{document}

\maketitle

\abstract{
The \textbf{stochastic block model} (SBM) is a generalization of the Erd\H{o}s--R\'enyi model of random graphs that describes the interaction of a finite number of distinct communities. In sparse Erd\H{o}s--R\'enyi graphs, it is known that a linear-time algorithm of Karp and Sipser achieves near-optimal matching sizes asymptotically almost surely, giving a law-of-large numbers for the matching sizes of such graphs in terms of solutions to an ODE \cite{ks}. We provide an extension of this analysis, identifying broad ranges of stochastic block model parameters for which the Karp--Sipser algorithm achieves near-optimal matching sizes, but demonstrating that it cannot perform optimally on general SBM instances. 

We also consider the problem of constructing a matching \textit{online}, in which the vertices of one half of a bipartite stochastic block model arrive one-at-a-time, and must be matched as they arrive. We show that the competitive ratio lower bound of 0.837 found by Mastin and Jaillet for the Erd\H{o}s--R\'enyi case \cite{er-online} is tight whenever the expected degrees in all communities are equal. We propose several linear-time algorithms for online matching in the general stochastic block model, but prove that despite very good experimental performance, none of these achieve online asymptotic optimality.
}{Stochastic Block Model, Matching Algorithms, Online Matching}

\section{Introduction}
A number of real-life allocation problems -- from organ donation to ridesharing or internet dating -- can be framed as matching problems. In many settings it is necessary to assign these matches very quickly without knowledge of future requests. This is particularly the case for the display of ads in search engine results: determining which ads to associate with the searched keywords must be done nearly-instantaneously. This problem has inspired a substantial line of research investigating the problem of online matching in random graphs.

Most prior work on this problem assumes the vertices on one side of a bipartite graph are drawn i.i.d.\ from an adversarially-chosen distribution. In that setting, upper bounds on the competitive ratio are known \cite{manshadi2011online}. However, in Erd\H{o}s--R\'enyi graphs, it is possible to exceed these bounds \cite{er-online}. One might therefore hope that graphs arising in nature tend to permit better online matching algorithms than adversarial distributions. In this work, we consider matching problems in the \defn{stochastic block model}, in which vertices belong to a constant number of ``classes'', and the probability of an edge between two vertices depends only on their classes. This is a broad class of structured distributions on graphs, which includes the Erd\H{o}s--R\'enyi model as a special case.

In stochastic block models, neither the optimal online matching algorithm nor the true matching number are known. We make progress towards both of these questions, finding expressions for the matching number in a number of regimes, and proposing and analyzing several heuristics for online matching.

\subsection{Preliminaries}

The idea of a population divided into a fixed number of distinct but internally-homogeneous groups is captured by the \textbf{stochastic block model}, first proposed by Holland et.\ al.\ to model social networks~\cite{HOLLAND}.

\begin{definition}[Stochastic block model]
    Consider $k$ disjoint sets of vertices (classes) $S_1, \dots, S_k$, and a symmetric probability matrix associating a value $p_{ij}$ to each pair $i,j$ of classes. Given these parameters, the \textbf{stochastic block model} (SBM) is the distribution over graphs obtained by adding each edge $(u,v)$ independently with probability $p_{\sigma_u\sigma_v}$, where $\sigma_u$ and $\sigma_v$ are the classes to which $u$ and $v$ respectively belong (the ``labels" of $u$ and $v$).
\end{definition}
If $k = 1$, we recover the Erd\H{o}s--R\'enyi graph $G(n,p)$ on $n$ vertices with edge probability $p$. If $k=2$, where each class has $n$ vertices, and $p_{12} = p$, $p_{11} = p_{22} = 0$, we recover the bipartite Erd\H{o}s--R\'enyi graph $G(n,n,p)$.
\medskip 

We are interested in the size of a maximum-cardinality matching on a graph drawn from this distribution.

\begin{definition}[Matching]
A \textbf{matching} of a graph is a subgraph $M$ such that no two edges in $M$ contain the same vertex; in other words, every vertex is matched with at most one other vertex via an edge. The \textbf{matching number} of a graph $G$ is defined as the size of the largest matching of $G$. A matching $M$ is called \textbf{perfect} if $|M| = |G|$.
\end{definition}

We are also interested in the problem of constructing a matching \textbf{online}.

\begin{definition}[Online bipartite matching problem] Suppose we have a bipartite instance of the stochastic block model, where there are $n$ ``left" vertices and $n$ ``right" vertices, the classes on the left have zero connection probability to the other left classes, and likewise the right classes have no connection probability to each other. In addition, classes are contained entirely within either the left or right. An ``online" algorithm is given the labels of all the right vertices, and knows a distribution over the left classes. For each of $n$ time-steps, a new left-node is revealed (a left label is drawn from the distribution, and coins are flipped to determine which edges it has to the right vertices) -- the algorithm then decides which if any of these edges to add to $M$. Once it has made this decision and moved on, it is never allowed to revisit that vertex.
\end{definition}

The most interesting range of stochastic block model parameters turns out to be the \defn{sparse regime}, when all probabilities $p_{ij}$ are $\Theta(1/n)$. The reason for this is because when $p_{ij}$ grows faster than $1/{n}$, as $n$ grows to infinity the graph becomes dense enough that with high probability, there is a near perfect matching between $S_i$ and $S_j$. On the other hand, when $p_{ij}$ grows slower than ${1}/{n}$, the graph becomes so sparse that we can include almost every edge in $M$. So, we consider the regime in which $p_{ij} = c_{ij}/n$ for some constants $c_{ij}$.

\subsection{Background}

In 1981, Karp and Sipser demonstrated that a simple linear-time heuristic achieves matchings within $o(n)$ of the true matching number of Erd\H{o}s--R\'enyi graphs~\cite{ks}. By associating the performance of that algorithm with a deterministic procedure, they were able to prove a law of large numbers on the matching number of such graphs. That paper has prompted continued investigation into their algorithm. In 1998, Frieze, Pittel, and Aronson improved the error estimate of the Karp--Sipser algorithm in Erd\H{o}s--R\'enyi graphs from $o(n)$ to $n^{1/5 + o(1)}$~\cite{ks-revisited}. In 2011, Bohman and Frieze extended analysis of the Karp--Sipser algorithm to the model of graphs drawn uniformly over a fixed degree sequence, showing that a log concavity condition is sufficient for the algorithm to find near-perfect matchings in such graphs \cite{bohman}. Because of its simplicity, the Karp--Sipser algorithm has also received attention as a practical method for data reduction; some recent work investigates efficient implementations \cite{kaya,shared-memory}.

The online bipartite matching problem was first introduced by Karp, Vazirani and Vazirani in 1990; they showed that a tight $1 - {1}/{e}$ competitive ratio is achievable on worst-case inputs \cite{karponline}. In 2009, Feldman et.\ al.\ showed that in the model where the left vertices are instead drawn from an arbitrary known distribution, with integral expected arrival rates, it is possible to get a competitive ratio strictly better than $1 - {1}/{e}$ \cite{feldmanonline}. The problem of online matching in the known distribution model has since seen considerable attention because of applications in internet ad allocation. For arbitrary distributions (with arbitrary arrival rates), the best known algorithm achieves competitive ratio 0.716, and there is a known upper bound of 0.823 \cite{huang2022power,manshadi2011online}. There has also been work considering algorithms for specific left vertex distributions. Mastin and Jaillet found that in $G(n,n,p)$, the random bipartite graph where all edges are independent and equally likely to exist, all greedy algorithms achieve competitive ratio at least 0.837~\cite{er-online}. Sentenac et.\ al.\ studied the problem in the 1-dimensional geometric model, where they found expressions for both the true matching size and the performance of a particular online heuristic \cite{sentenac2023online}. To the best of our knowledge, the only previous work that considers the stochastic block model is by Soprano-Loto et.\ al., who consider the regime where the graph is dense (i.e.\ all probabilities are constants not depending on $n$), and characterize when it is possible to achieve an asymptotically near-perfect matching  \cite{sopranoloto2023online}.

\subsection{Main Contributions}

We show that the Karp--Sipser algorithm achieves near-optimal matchings for probability matrices $p_{ij}$ satisfying any of the following list of conditions. However, the algorithm
does not achieve near-optimal matchings in general stochastic block models.
\begin{itemize}
\item \textbf{Equitable}: For every $i$, $\sum_j {p_{ij} |S_j|} = \sum_j {c_{ij} |S_j|}/{n} = c$ for some constant $c$, i.e., each vertex has the same expected degree, we are in the ``equitable" case. Here, we show that the asymptotic matching number for any such graph is $\alpha n + o(n)$ where $\alpha$ is an explicit constant.
See \cref{equitable_offline} for the full statement.

This result was previously known for the sparse Erd\H{o}s--R\'enyi graph $G(n, c/n)$~\cite{ks}. Equitable block models capture in particular the standard symmetric $p_{in}$-$p_{out}$ model, in which the population is divided into equal-sized groups, and there are distinct probabilities for connecting within-group as opposed to across-group.
\item \textbf{Sub-Critical}: When $\sum_j c_{ij}|S_j|/n < e$ for all $i$, we show that the model is in a sub-critical regime similar to the one found for Erd\H{o}s--R\'enyi graphs. We show that the asymptotic matching number converges to the solution of an explicit ODE, see \cref{thm:subcritical}.
\item \textbf{Bipartite Erd\H{o}s--R\'enyi}: We also determine, in terms of the solution of an explicit ODE, the asymptotic matching number of $G(kn,n,c/n)$, the bipartite graph with part sizes $kn$ and $n$ and independent edge probability $c/n$; see \cref{lem:bipartite-er-offline}. This case is of interest particularly as a simple example for which usual arguments about the Karp--Sipser algorithm fail.
\end{itemize}
As a special case, these results imply that Mastin and Jaillet's upper bound on the matching number of $G(n, n, p)$ is tight, and so their 0.837 competitive ratio lower bound for the online version is also tight. We propose a modification of the Karp--Sipser algorithm that we conjecture to be asymptotically optimal in general, but observe that known approaches do not suffice to prove optimality.

With regards to the online matching problem, we analyze the following heuristics:
\begin{itemize}
    \item \dumb: When a match is possible, match to a uniform random neighbour.
    \item \degreedy: Match to the available class with the lowest expected degree.
    \item \shortsighted: Match to the available class that maximizes the probability of being able to match on the next step.
    \item \bruteforce: Precompute the expected final matching size of every possible state, and then always match to the available class that maximizes the expected size (this gives optimal performance among all online algorithms, but is very inefficient).
\end{itemize}
In the equitable case, we find that all of these algorithms achieve the same tight 0.837 competitive ratio as in $G(n, n, p)$. In the general setting, we show that only \bruteforce achieves optimal asymptotic performance.

\section{Analysis of the Karp--Sipser Algorithm for Offline Matching}\label{sec:offline}
In 1981, Karp and Sipser proposed an algorithm for approximating the matching number of graphs, and proved that it achieves near-optimal matching sizes on Erd\H{o}s--R\'enyi graphs \cite{ks}. By extending the analysis of their algorithm to the stochastic block model case, we show a law of large numbers on the matching size of a more general class of graph.

\subsection{Karp--Sipser Algorithm and Outline of Analysis}
The form of the Karp--Sipser algorithm we study is as follows: 
\begin{algorithm}[H]
\caption{Karp--Sipser}
\begin{algorithmic}[1]
  \State $M \gets \emptyset$
  \State $V \gets V(G)$
  \While{$E(G \cap V) \neq \emptyset$}
    \If{exists $v \in V$ of degree 1 -- that is, such that exactly one $u \in V$ has $(uv) \in E(G)$}
        \State Choose such a $v$ uniformly at random over all degree 1 vertices
        \State $M \gets M \cup \{(uv)\}$
        \State $V \gets V \setminus \{u, v\}$
    \Else 
        \State Choose an edge $(uv)$ uniformly randomly over all $u,v \in V$, $(uv) \in E(G)$
        \State $M \gets M \cup \{(uv)\}$
        \State $V \gets V \setminus \{u, v\}$
    \EndIf
  \EndWhile
  \State Return $M$.
\end{algorithmic}
\end{algorithm}

In other words, whenever there exists a vertex of degree 1 in the graph, choose one randomly, add its edge to the matching, and remove both endpoints. When there are no degree 1 vertices, repeatedly choose an edge at random to add to the matching. When there exist degree 1 vertices it's always ``safe" to add their edges in the sense that there always exists an optimal matching which does so, and so any ``mistakes" the algorithm makes can only happen after the first time it reaches 0 degree 1 vertices (we call the steps before this ``Phase 1", and the steps after this ``Phase 2"). The analysis of this algorithm in the Erd\H{o}s--R\'enyi graph setting proceeds as follows~\cite{ks, ks-revisited}:

\begin{enumerate}
    \item Conditioned on the state of a Markov process on small tuples of integers, the graph maintains a simple distribution law even after several steps of the algorithm.
    \item Use estimates of the degree distribution of the graph to determine transition probabilities for the Markov process.
    \item Appeal to known approximation theorems to conclude that the Markov process stays close to the solution of a corresponding ODE as $n \to \infty$ with high probability.
    \item Observe that, in Phase 2 of the algorithm, it is very likely that the algorithm finds a near perfect matching on the remaining graph. Since the algorithm makes only optimal decisions in Phase 1, this means that overall it finds within $o(n)$ vertices of the true optimal value, and so the matching number of the graph is described by solutions to the ODE from step (iii).
\end{enumerate}

We apply steps analogous to (i)-(iii) in the general stochastic block model. On the other hand, step (iv) of the analysis is not true in general; although we list some interesting cases in which it is.

\subsection{Transition Probabilities of the Karp--Sipser Algorithm in Configuration Models} 
For analytical purposes, we instead consider the Karp--Sipser algorithm on multigraphs drawn from a configuration model.

\begin{definition}
[Blocked configuration model] Let $G$ be a graph drawn from a stochastic block model. For each pair of labels $i,j$, let $m_{ij}$ denote the total number of edges between vertices of label $i$ and vertices of label $j$ in $G$. Construct a random multigraph as follows: if $i \neq j$, distribute $m_{ij}$ half-edges among the $i$ vertices and $m_{ij}$ half-edges among the $j$ vertices uniformly at random (à la balls-in-bins), then define edges by a uniform random pairing between the half edges on the $i$ and $j$ sides. If $i = j$, instead distribute $2m_{ii}$ half-edges within class $i$, then choose a uniform random pairing among those $2m_{ii}$ half-edges. This process defines a distribution over multigraphs, because in the process of ``reshuffling" the half-edges of $G$ we introduce the possibility of multiple edges and self-loops.
\end{definition}

Equivalently, once the $m_{ij}$ are determined, the process chooses, uniformly at random for each $i, j$, an ordered list of $m_{i,j}$ many edges from $i$ to $j$ (with possible duplicates). We show in \cref{appendix:configuration} that any result which holds with high probability in the blocked configuration model also holds with high probability for the stochastic block model.

As we run the Karp--Sipser algorithm on a random multigraph, the distribution changes, since the earlier steps of the algorithm have produced some effects conditional on the previous states of the graph. Fortunately, it turns out that these effects have a simple description.
\begin{lemma}\label{lem:markov-holds}
Suppose we generate a graph from the blocked configuration model and run the Karp--Sipser algorithm for an arbitrary number of steps. Then conditioned on the following, the resulting graph is still distributed according to the blocked configuration model.
\begin{itemize}
    \item For each pair of $(i, k)$ of block model labels, the number $E_{ij}$ of edges between label class $i$ and $j$ (when $i = j$, we let $E_{ij}$ denote \emph{twice} the number of edges lying within class $i = j$, so that $E_{ij}$ always refers to the number of $j$-type half-edges attached to label-$i$ vertices).
    \item For each label class $i$, the number $T_i$ of label-$i$ vertices of degree exactly 1 (the ``thin" vertices).
    \item For each label class $i$, the number $F_i$ of label-$i$ vertices of degree at least 2 (the ``fat" vertices).
\end{itemize}
Thus, if we collect all of these values into a tuple $Y = (E_{ij}, T_i, F_i)$, the algorithm's progress can be described by a Markov chain on $Y$. 
\end{lemma}
A proof of \cref{lem:markov-holds} is given in \cref{appendix:markov}. In order to determine transition probabilities of this Markov process, we note the following straightforward lemma about the degree distribution, justified in \cref{appendix:degrees}:

\begin{lemma}\label{lem:degreedist}
Fix a description tuple $Y$. Let $nY$ be the tuple obtained by scaling every element of $Y$ by $n$, and let $G_n$ be a multigraph drawn from a blocked configuration model conditional on tuple $nY$. Choose a uniform half-edge incident to class $i$, and let $v$ be its incident vertex. In the limit as $n \rightarrow \infty$, the degree of $v$ converges in distribution to 
\[
\mathbb{P}[d(v) = k] = \begin{cases}
\frac{T_i}{\sum_l E_{il}} & \text{ for $k = 1$} \\
\frac{\left(\sum_l E_{il}\right) - T_i}{\sum_l E_{il}} \cdot \frac{\lambda^{k-1}}{(k-1)! e^\lambda \left( 1 - e^{-\lambda}\right)} & \text{ for $k > 1$},
\end{cases}
\]
where $\lambda$ is a solution to 
\[\frac{\lambda e^{\lambda} - \lambda}{e^{\lambda}-1-\lambda} = \frac{\left(\sum_l E_{il}\right) - T_i}{F_i}.\]
\end{lemma}

This allows us to determine the probability that a random edge is adjacent to a degree 2 vertex, which is important for understanding the evolution of the algorithm.
\begin{corollary}
    In particular, as $n \rightarrow \infty$, the probability that a random edge into class $i$ connects to a vertex of degree exactly 2 tends to $ \frac{\left(\sum_l E_{il}\right) - T_i}{\left(\sum_l E_{il}\right)} \cdot \frac{\lambda}{e^\lambda - 1}$. 
\end{corollary}

The above Lemma also allows us to determine the label distribution in the neighborhood of a randomly chosen edge.  
\begin{corollary}
    Let $uv$ be a randomly chosen edge where $u$ has label $i$ and $v$ has label $j$. As $n \rightarrow \infty$, the number of neighbors of $u$ with label $k$ tends to $\frac{\left(\sum_l E_{il}\right) - T_i}{\sum_l E_{il}} \cdot \frac{\lambda}{1-e^{-\lambda}}\cdot\frac{E_{ik}}{ \left(\sum_l E_{il}\right)  }$. 
\end{corollary}
\begin{proof}
    The above degree estimates tell us that there are $\frac{\left(\sum_l E_{il}\right) - T_i}{\sum_l E_{il}} \cdot \frac{\lambda}{1-e^{-\lambda}}$ other edges out of $u$ in expectation. Because half-edges of all types are distributed interchangeably, the fraction of them that go to label $k$ is $\frac{E_{ik}}{ \left(\sum_l E_{il}\right) }$.
\end{proof}

From these degree distribution estimates, we produce estimates on the transition probabilities of the Markov process. On each step of the algorithm, if there exist degree 1 vertices, the algorithm chooses one of them and removes it and its neighbour. So, the graph loses one edge from between class $i$ and $j$ whenever that degree 1 vertex has its edge between $i$ and $j$. When the neighbour vertex is in class $i$, the graph also loses edges equal to however many neighbours it had in class $j$. Similar accounting can be made for the number of fat or thin vertices in a class: the graph loses a fat vertex either by having it as the neighbour of the degree 1 vertex we removed, or by having it initially have degree 2 and appear as the neighbour of a removed neighbour, so that it's then reduced to degree 1. A thin vertex is lost whenever its the removed degree 1 vertex, whenever its the neighbour of the removed vertex, or whenever its a neighbour of the neighbour of a removed vertex, but gain one whenever a degree 2 vertex is a neighbour-of-the-neighbour. To make these expressions explicit, we define the notation:

\begin{itemize}
    \item Let $h_i$ denote the total number of half-edges in class $i$,
    \[h_i = \sum_l E_{il}.\]
    \item Let $\omega_{ij}$ denote the probability that a randomly selected degree-1 vertex is in class $i$ and has its neighbour in class $j$, 
    \[\omega_{ij} =
    \frac{T_i}{\sum_{l} T_l} \cdot \frac{E_{ij}}{h_i}.\]
    \item Let $\delta_{ij}$ denote expected number of other $j$-type half-edges attached to the vertex a random half-edge in class $i$ is attached to. If $\lambda$ is a solution to $\frac{\lambda(e^\lambda - 1)}{e^\lambda - 1 - \lambda} = \frac{h_i - T_i}{F_i}$, then, by our degree estimates,
    \[\delta_{ij} = \frac{h_i - T_i}{h_i} \cdot \frac{\lambda}{1-e^{-\lambda}}\cdot \frac{E_{ij}}{h_i}.\]
    \item Let $\theta_i$ denote the probability that a random half-edge attached to class $i$ is attached to a degree 2 vertex. Again, letting $\frac{\lambda(e^\lambda - 1)}{e^\lambda - 1 - \lambda} = \frac{h_i - T_i}{F_i}$, we have 
    \[\theta_i = \frac{h_i - T_i}{h_i} \cdot \frac{\lambda}{e^\lambda - 1} .\]
\end{itemize}

Now, while there are degree-1 vertices remaining in the graph, we can write the expected change in the description tuple after one step of the algorithm as
\begin{align*}
    \mathbb{E}[\Delta E_{ij}] &= - \omega_{ij} - \omega_{ji} - \sum_{l} \omega_{li}\delta_{ij} - \sum_{l} \omega_{lj}\delta_{ji} \\
    \mathbb{E}[\Delta F_i] &= -\left(\sum_l \omega_{li}\right)\left(\frac{h_i - T_i}{h_i}\right) - \sum_j\sum_l \omega_{jl}\delta_{li} \theta_i \\
    \mathbb{E}[\Delta T_i] &= - \left( \sum_l \omega_{il} \right) - \left( \sum_l \omega_{li} \right)\left(\frac{T_i}{h_i}\right) - \sum_j \sum_l \omega_{jl} \delta_{li} \left( \frac{T_i}{h_i} - \theta_i \right).
\end{align*}

On the other hand, when there are no degree 1 vertices remaining, we choose an edge uniformly at random, so, by similar reasoning,
\begin{align*}
    \mathbb{E}[\Delta E_{ij}] &= - \left(\frac{2 E_{ij}}{\sum_k h_k}\right) - \left(\frac{2 h_i}{\sum_k h_k}\right)\delta_{ij} - \left(\frac{2 h_j}{\sum_k h_k}\right)\delta_{ji} \\
\mathbb{E}[\Delta F_i] &= - \left(\frac{2 h_i}{\sum_k h_k}\right) - \sum_l \left(\frac{2 h_l}{\sum_k h_k}\right)\delta_{li} \theta_i
\\
\mathbb{E}[\Delta T_i] &=  \sum_l \left(\frac{2 h_l}{\sum_k h_k}\right)\delta_{li} \theta_i.
\end{align*}
We now argue that the evolution of Phase 1 of the algorithm stays close to the solution of an ODE.

\subsection{Convergence to Continuous Approximation \label{phase1ode}}
Associating Markov processes on graphs to differential equations is a very useful tool, and there have been a variety of versions of this argument with varying levels of generality and probability bound guarantees \cite{kurtz,darling-norris}. Here, we justify the passage to differential equations by Wormald's theorem \cite{wormald}, although we relegate the verification of the technical conditions of the theorem to \cref{appendix:wormwald}. The important thing to note is that the expected transitions above are ``scale-invariant", meaning that they remain the same upon re-scaling all entries in $Y$ by the same amount. So, letting $\bar{E}_{ij} = {E_{ij}}/{n}$, $\bar{T}_{i} = {T_i}/{n}$, $\bar{F}_i = {F_i}/{n}$, we can write (for Phase 1):
\begin{align*}
\mathbb{E}[\Delta \bar{E}_{ij}] &= \frac{- \omega_{ij} - \omega_{ji} - \sum_{l} \omega_{li}\delta_{ij} - \sum_{l} \omega_{lj}\delta_{ji}}{n} \\
\mathbb{E}[\Delta \bar{F}_i] &= -\frac{\left(\sum_l \omega_{li}\right)\left(\frac{\bar{h}_i - \bar{T}_i}{\bar{h}_i}\right) - \sum_j\sum_l \omega_{jl}\delta_{li} \theta_i}{n} \\
\mathbb{E}[\Delta \bar{T}_i] &= - \frac{\left( \sum_l \omega_{il} \right) - \left( \sum_l \omega_{li} \right)\left(\frac{\bar{T}_i}{\bar{h}_i}\right) - \sum_j \sum_l \omega_{jl} \delta_{li} \left( \frac{\bar{T}_i}{\bar{h}_i} - \theta_i \right)}{n}.
\end{align*}

This is a process that takes order $n$ time steps, and where the expected change at each time step scales like $\frac{1}{n}$. Informally, we can observe that in the limit of $n$, the many small steps should average out and produce a process evolving according to their expectations; this suggests looking at the following system of equations:
\begin{align*}
    \frac{d}{dt} \bar{E}_{ij}(t) &= - \omega_{ij}(t) - \omega_{ji}(t) - \sum_{l} \omega_{li}(t)\delta_{ij}(t) - \sum_{l} \omega_{lj}(t)\delta_{ji}(t) \\ 
    \frac{d}{dt} \bar{F}_i(t) &= -\left(\sum_l \omega_{li}(t)\right)\left(\frac{\bar{h}_i(t) - \bar{T}_i(t)}{\bar{h}_i(t)}\right) - \sum_j\sum_l \omega_{jl}(t)\delta_{li}(t) \theta_i(t) \\ 
    \frac{d}{dt} \bar{T}_i(t) &= - \left( \sum_l \omega_{il}(t) \right) - \left( \sum_l \omega_{li}(t) \right)\left(\frac{\bar{T}_i(t)}{\bar{h}_i(t)}\right) - \sum_j \sum_l \omega_{jl}(t) \delta_{li}(t) \left( \frac{\bar{T}_i(t)}{\bar{h}_i(t)} - \theta_i(t) \right)
\end{align*}
with initial conditions
\[\bar{E}_{ij}(0) = c_{ij} \bar{S}_i\bar{S}_j, \ \  
\bar{F}_i(0) = 1 - \Big( 1 + \sum_j c_{ij} \bar{S}_j \Big) e^{- \sum_j c_{ij} \bar{S}_j } \ \text{ and }  \
\bar{T}_i(0) = \Big( \sum_j c_{ij} \bar{S}_j  \Big) e^{- \sum_j c_{ij} \bar{S}_j },\]
where $c_{ij}$ and $\bar{S}_i = {|S_i|}/{n}$ are the connection probabilities and label class sizes, respectively, of the stochastic block model instance. Wormald's theorem guarantees that, in the limit of $n$, the evolution of Phase 1 stays close to the unique solution of this ODE with probability approaching 1 (this is formally justified in \cref{appendix:wormwald}). This implies the following:

\begin{lemma}
    If $\mathcal{Y}(t) = \{\mathcal{\bar{E}}_{ij}, \mathcal{\bar{F}}_{i}, \mathcal{\bar{T}}_{i} \}$ is a solution to the above ODE, with high probability, the total number of unmatched isolated vertices created in Phase 1 of the Karp--Sipser algorithm is $n \left( 1 - \tau - \mathcal{\bar{F}}_i(\tau)\right) + o(n)$, where $\tau$ is the first time such that $\mathcal{\bar{T}}_i(\tau) = 0$ for all $i$.
\end{lemma}

These equations don't seem to have a simple analytical solution in general, but they can be effectively numerically evaluated for any specific stochastic block model instance. In the Erd\H{o}s--R\'enyi case, studying Phase 2 of the Karp--Sipser algorithm reveals that with high probability at most $o(n)$ unmatched isolated vertices are created, meaning that the algorithm is asymptotically optimal, and that the true matching size is $n \left(\tau + \mathcal{\bar{F}}_i(\tau)\right) + o(n)$ with high probability. However, this analysis turns out not to work for general stochastic block model instances. We first illustrate a few examples (namely the equitable and sub-critical cases) where, with a little bit of work, we can prove similar results; then, we examine where the algorithm fails.

\subsection{Equitable Case}

The first, and most important success story we show for the Karp--Sipser algorithm happens in a case we denote ``equitable":

\begin{definition}
    We call stochastic block model parameters \textbf{equitable} if there is some constant $c$ such that for all classes $i$, 
    \[\sum_j \frac{c_{ij} |S_j|}{n} = c.\]
\end{definition}
In other words, although the edge density in some parts of the graph may be higher than other parts, the expected degree of every vertex is $c$ regardless of what label class it belongs to. In these cases, we show that not only does the Karp--Sipser algorithm construct an asymptotically-optimal matching, but that the matching size it constructs is asymptotically the same as the matching number of the Erd\H{o}s--R\'enyi graph $G(n, c/n)$. 
The intuition behind this claim is that, despite the nontrivial correlation between the edges of the graph, we expect the degree distributions to look the same everywhere. So, since our estimates of transition probabilities were functions of the degree distributions, they should evolve similarly to the Erd\H{o}s--R\'enyi case. The crucial point we need to justify to make this intuition precise is that the degree distributions necessarily \textit{remain} close to equal across classes, given that they start that way.

\begin{theorem}\label{equitable_offline}
    With high probability, the matching number of an equitable stochastic block model is 
    \[\left(1 - \frac{x + ce^{-x} + xce^{-x}}{2c}\right) n + o(n),\]
    where $x$ is the smallest solution to $x = ce^{-ce^{-x}}$, and the Karp--Sipser algorithm achieves within $o(n)$ of this value.
\end{theorem}
The theorem is proven with the help of the following two lemmas.

\begin{lemma}\label{lem:equitable-phase1}
Given an equitable stochastic block model, after Phase 1 of the Karp--Sipser algorithm, with high probability for every $i$, $F_i$ (number of vertices) and $h_i$ (total number of incident edges) differ by at most $o(n)$ from the corresponding values in $G(n, \frac{c}{n})$.
\end{lemma}

\begin{proof}
This follows directly from our application of Wormald's theorem. We know that initially all ${\bar{T}_i}/{\bar{S}_i}$, ${\bar{F}_i}/{\bar{S}_i}$'s and ${\bar{h}_i}/{\bar{S}_i}$'s are equal; now observe that if that equality holds at some time $t$, we have
\begin{align*}
    \frac{d}{dt} \bar{h}_i(t) &= \sum_j \frac{d}{dt} \bar{E}_{ij}(t) = \sum_j \left( - \omega_{ij}(t) - \omega_{ji}(t) - \sum_{l} \omega_{li}(t)\delta_{ij}(t) - \sum_{l} \omega_{lj}(t)\delta_{ji}(t) \right) = \bar{S}_i \left( - 2 - 2\delta \right) \\
    \frac{d}{dt} \bar{F}_i(t) &= -\left(\sum_l \omega_{li}(t)\right)\left(\frac{\bar{h}_i(t) - \bar{T}_i(t)}{\bar{h}_i(t)}\right) - \sum_j\sum_l \omega_{jl}(t)\delta_{li}(t) \theta_i(t) = \bar{S}_i \left( \frac{\bar{h} - \bar{T}}{\bar{h}} - \delta \theta \right) \\
    \frac{d}{dt} \bar{T}_i(t) &= - \left( \sum_l \omega_{il}(t) \right) - \left( \sum_l \omega_{li}(t) \right)\left(\frac{\bar{T}_i(t)}{\bar{h}_i(t)}\right) - \sum_j \sum_l \omega_{jl}(t) \delta_{li}(t) \left( \frac{\bar{T}_i(t)}{\bar{h}_i(t)} - \theta_i(t) \right) = \bar{S_i} \left( \frac{\bar{T}}{\bar{h}} - \delta \left( \frac{\bar{T}}{\bar{h}} - \theta \right) \right),
\end{align*}
where $\bar{h} = {\bar{h}_i}/{\bar{S}_i}$, $\bar{T} = {\bar{T}_i}/{\bar{S}_i}$, $\delta$ is our estimate on the expected degree - 1 of the vertex attached to a random half-edge, and $\theta$ is our estimate on the probability the degree of the vertex attached to a random half-edge is exactly 2 (it can be seen from our estimates that these values are the same across classes). Note that these values evolve the same way in every class, just scaled by the size of the class.

So, in a solution to these equations, the values of ${\bar{T}_i}/{\bar{S}_i}$, ${\bar{F}_i}/{\bar{S}_i}$ and ${\bar{h}_i}/{\bar{S}_i}$ remain equal across classes for all time, and the evolution of the total number of edges, thin vertices, and fat vertices follows the same trajectory as it would in $G(n, c/n)$. Since Wormald's theorem guarantees that with high probability Phase 1 differs by at most $o(n)$ from a solution to this equation, we have the desired statement. \end{proof}

\begin{lemma}\label{lem:equitable-phase2}
At the start of Phase 2 of the Karp--Sipser algorithm, if the values of ${F_i}/{\bar{S}_i}$ and ${h_i}/{\bar{S}_i}$ each differ by at most $o(n)$ between classes, then with high probability, at most $o(n)$ isolated vertices are produced in Phase 2.
\end{lemma}

\begin{proof}

Fix some constant $\epsilon > 0$. We want to show that for any such $\epsilon$, the probability of isolating more than $\epsilon n$ vertices over the course of Phase 2 tends to 0 as $n \to \infty$. We divide our analysis into two parts:
\begin{enumerate}
    \item First, we consider the case where the average degree ${h_i}/{F_i} > 2 + \epsilon$ for all classes $i$. In this regime, we argue that the number of steps between the times when the graph is free of thin vertices is small, so we can again control the evolution by an ODE. The following is a rough sketch of the argument (a more detailed justification of these points involving Wormald's theorem is given in \cref{appendix:wormwald}):
    
    When we first remove a random edge, this may create some vertices of degree 1. In removing those, we may create more vertices of degree 1. In general, the expected number of new degree-1 vertices created when a degree-1 vertex of class $i$ is removed is $\sum_j \frac{E_{ij}}{h_i} \delta_i \theta_j \leq \left(\max_i \delta_i\right) \left(\max_j \theta_j \right)$. Now, when we have that the average degree in each class is at least $2 + \epsilon$, and we know that the difference in average degree between any pair of classes $i$ and $j$ is very small (say, less than $\gamma$), then we know $\delta_i \theta_j = \frac{\lambda_i }{1 - e^{-\lambda_i}} \cdot \frac{\lambda_j}{e^{\lambda_j} - 1} < 1 - \eta$ for any pair $i$, $j$ of classes, where $\eta$ depends on $\gamma$ and $\epsilon$ (the existence of some $\gamma > 0$, $\eta > 0$ in terms of $\epsilon$ with this property is guaranteed by a continuity argument). So, the expected number of degree-1 vertices created for each degree-1 vertex removed is at most a constant, $1 - \eta$, that is less than 1. The size of a subcritical Galton--Watson tree with $1 - \eta$ expected offspring is very unlikely to exceed a bound in terms of $\eta$. Thus, while we are within the region where average degrees are $\gamma$-close and greater than $2 + \epsilon$, the duration of a "run" of degree-1 stripping is unlikely to be much more than a constant independent of $n$. Since the length of each run is small, we can appeal to the law of large numbers and claim the process evolves like its expectation. Whenever the average degree is the same in all classes, the expected change in the number of edges of a given type is proportional to the number of edges currently of that type (all edges in the graph are equally likely to be chosen as the first edge removed, equally likely to be the edge chosen one step into the run, etc). So, since average degrees start out $o(1)$ away from each other, we expect them to remain that way up until one of them drops below $2 + \epsilon$. Also, for a run of constant length we expect to create $o(1)$ isolated vertices, so the total number created while we're in this regime is $o(n)$ with high probability.
    
    \item Once we've left that region and the average degree in some class has dropped below $2 + \epsilon$, by the above analysis the average degree in \textit{every} class is $o(1)$ away from $2 + \epsilon$. If we perform another run of the algorithm, we're no longer guaranteed that the expected change in number of degree-1 vertices is bounded away from 1 -- we do know, however, that the only way to remove a thin vertex and create more than one in its place is to have its neighbour have degree greater than 2. Since no class has average degree more than $2 + 2\epsilon$, we know that the entire graph has at most $2n \epsilon$ edges associated with vertices of degree greater than 2. Those edges are the only places we can branch out and create more degree-1 things than we consume, so, throughout the course of the rest of the algorithm, there can never be more than $2 n \epsilon$ thin vertices in the graph at once. 
    
    Now, note that the rate of creation of isolated vertices on a given step is always proportional to the fraction of thin vertices in the graph: when we remove a vertex (the neighbour of a thin vertex), it has $\delta_i < 2$ other neighbours in expectation, and by our Markov property we know that those neighbour edges are equally likely to be any of the edges leaving the class. So, the number of vertices isolated at each step is in expectation at most $\frac{2}{v}$ times the total number of degree-1 vertices, where $v$ is the total number of vertices in the graph. By another law-of-large-numbers argument, when we perform $\Theta(n)$ such steps, the probability of being a total of $\Theta(n)$ above this expected bound is $o(1)$. So, we can upper bound the number of isolated vertices by 
    \[ 2 \epsilon n + \int_{2 \epsilon n}^n \frac{2}{v} 2 \epsilon n dv = 2 \epsilon n \left( 1 - 2 \log(2 \epsilon) \right). \]
    We can make this arbitrarily small by choice of $\epsilon$; so, with high probability the total number of vertices isolated in Phase 2 is $o(n)$. \qedhere
\end{enumerate}
\end{proof}
\begin{proof}[Proof of \cref{equitable_offline}]
This follows directly from Lemmas~\ref{lem:equitable-phase1} and \ref{lem:equitable-phase2}. In Phase 1, the Karp--Sipser algorithm is guaranteed to perform optimally, and in Phase 2, with high probability, it isolates only $o(n)$ vertices, so the algorithm is within $o(n)$ of optimal. Since the ODE determining the evolution of Phase 1 evolves the same as in an Erd\H{o}s--R\'enyi graph with parameter $c/n$, the total number of lost vertices must be within $o(n)$ of the number lost in the Erd\H{o}s--R\'enyi case. Since the above expression is known to be the matching number in the Erd\H{o}s--R\'enyi case \cite{ks,ks-revisited}, it must therefore also be here. \end{proof}

\subsection{Sub-Critical case}

In the Erd\H{o}s--R\'enyi case, Karp and Sipser proved that the number of unmatched non-isolated vertices remaining in the graph after Phase 1 (which we follow recent literature in calling the ``Karp--Sipser core" \cite{budzinski2022critical}) is $o(n)$ with high probability if $c < e$, and $\Theta(n)$ with high probability if $c > e$ \cite{ks}. We have seen that, even when we allow some correlation into the graph in terms of different classes, so long as the expected degree into all vertices is the same across classes, the progress of the algorithm is the same as if the graph was Erd\H{o}s--R\'enyi. This implies that any equitable stochastic block model also follows this critical transition at $c = e$. In this section, we examine criticality in the non-equitable case -- can we describe a similar phase transition in terms of the expected degrees of a general block model? We might initially expect a statement of the following form:

\begin{claim}[False] The Karp--Sipser core of a graph drawn from a stochastic block model is $o(n)$ with high probability (that is, Phase 1 strips away all but a sublinear number of vertices) \textbf{if and only if} the expected degree $\sum_j c_{ij} \bar{S_j}$ is less than $e$ for all classes $i$.
\end{claim}

We prove the ``if" direction of this claim. However, the ``only if" direction turns out not to be true -- in fact, it appears to be possible for the model to be subcritical even when $\sum_j c_{ij} \bar{S_j} > e$ for all $i$. A complete characterization of the critical boundary is left as an open question.

\begin{theorem}\label{thm:subcritical}
If $\sum_j c_{ij} \bar{S_j} < e$ for all classes $i$, then the size of the Karp--Sipser core is $o(n)$.
\end{theorem}

In principle, analysis of the criticality of the Karp--Sipser algorithm could be done by analysis of the ODE from \cref{phase1ode}; however, instead of trying to understand that system in general, we follow the methodology of Karp and Sipser's original paper in associating the probability of removing a vertex in the graph with the probability of removing the root of a random tree~\cite{ks}. We note the following facts:

\begin{fact}[Karp and Sipser~\cite{ks}]
    The set of vertices removed by Phase 1 is fixed, regardless of the order in which degree-1 vertices are stripped. So, if there exists some valid sequence of degree-1 strippings that removes a given vertex $v$ from the graph (either matching or isolating it), that vertex is not in the Karp--Sipser core.
\end{fact}

\begin{fact}[Implied by a result of Mossel, Neeman and Sly~\cite{mossel2012stochastic} for 2 classes; Sly and Chin \cite{chin2021optimal} for greater than 2 classes]
    For any constant $d$, in the limit of $n$ the $d$-neighbourhood of any vertex of $G$ (i.e., the subgraph obtained by a BFS of depth $d$ from the vertex) converges in distribution to the first $d$ levels of a multitype branching process, where nodes of type $i$ have independently Pois($c_{ij} \bar{S_j}$) children of type $j$, and the root class corresponds to the class of the vertex in $G$.
\end{fact}

If we can show that, under certain conditions, with probability tending to $1$ with $d$, there exists a sequence of valid Karp--Sipser vertex removals in this tree, all of which are at depth at most $d$, and which result in the root being removed, this then implies that the Karp--Sipser algorithm is subcritical. This is because we know that, in the limit of $n$, any structure that appears in the first $d$ levels of the tree is equally likely to appear in the $d$-neighbourhood of a given vertex in $G$; so, if with probability at least 1 - $\epsilon$ there is a way to remove the root of such a tree for any root class, the expected number of vertices remaining in $G$ after Phase 1 is at most $\epsilon n$. The following characterizes these conditions:
\begin{lemma}\label{lem:fixed_points_are_what_we_want}
    The probability of removing the root of this branching process tends to $1$ in the limit of $d$ whenever the following system has no more than one solution $(x_1, \dots, x_q)$ in $[0,1]^q$:
\[x_i = e^{- \left( \sum_j c_{ij} \bar{S}_j e^{- \left(\sum_k c_{jk} \bar{S}_k x_k\right)}\right)}.\]
\end{lemma}

The proof of \cref{lem:fixed_points_are_what_we_want} is given in \cref{appendix:games}, where it is derived as a corollary of results about winning probabilities of games on multitype branching processes (those results are an extension of work of Holroyd and Martin on Galton--Watson trees \cite{games}, and may be of independent interest). We now prove the forward (true) direction of our claim.

\begin{proof}[Proof of \cref{thm:subcritical}]
By the above analysis, it suffices to show that whenever $\sum_j c_{ij} \bar{S_j} < e$ for all classes $i$, the function 
\[\begin{bmatrix} x_1 \\ \dots \\ x_q \end{bmatrix} \mapsto \begin{bmatrix} e^{- \left( \sum_j c_{1j} \bar{S}_j e^{- \left(\sum_k c_{jk} \bar{S}_k x_k\right)}\right)} \\ \dots \\ e^{- \left( \sum_j c_{qj} \bar{S}_j e^{- \left(\sum_k c_{jk} \bar{S}_k x_k\right)}\right)} \end{bmatrix} - \begin{bmatrix} x_1 \\ \dots \\ x_q \end{bmatrix}\]
has only one root on $[0,1]^q$. Denoting $e^{- \left( \sum_j c_{ij} \bar{S}_j e^{- \left(\sum_k c_{jk} \bar{S}_k x_k\right)}\right)}$ as $f_i$, the Jacobian of this function looks like 
\[
\begin{bmatrix}
f_1 \left(\sum_j \left(c_{1j} \bar{S}_j\right)\left(c_{j1}\bar{S}_1 \right) e^{-\left(\sum_k c_{jk} \bar{S}_k x_k\right)} \right) - 1 & \dots & f_1 \left(\sum_j \left(c_{1j} \bar{S}_j\right)\left(c_{jq}\bar{S}_q \right) e^{-\left(\sum_k c_{jk} \bar{S}_k x_k\right)} \right) \\
\dots & \dots & \dots \\
f_q \left(\sum_j \left(c_{qj} \bar{S}_j\right)\left(c_{j1}\bar{S}_1 \right) e^{-\left(\sum_k c_{jk} \bar{S}_k x_k\right)} \right)& \dots & f_q\left(\sum_j \left(c_{qj} \bar{S}_j\right)\left(c_{jq}\bar{S}_q \right) e^{-\left(\sum_k c_{jk} \bar{S}_k x_k\right)} \right) -1\\
\end{bmatrix}.
\]
The sum of the entries in the $i$th row of this matrix is 
\[f_i \left( \sum_j \left( c_{ij} \bar{S}_j e^{-\left(\sum_k c_{jk} \bar{S}_k x_k\right)} \cdot \sum_l \left(c_{jl} \bar{S}_l \right) \right)\right) - 1.\]
By assumption, we know $\sum_l \left(c_{jl} \bar{S}_l \right) < e$. So, the above expression is strictly less than 
\[e f_i \log f_i - 1.\]

For any value of $f_i$, $e f_i \log f_i$ is at most 1 (taking the derivative, we find a unique maximum at $f_i = {1}/{e}$). So, we have shown that the sum of every row of the Jacobian is negative everywhere. Now, suppose that this function has two distinct roots, $x = (x_1, \dots, x_q)$ and $y = (y_1, \dots, y_q)$. Let $i$ be the index where $y_i - x_i$ is maximal. We have increased $x_i$ by $(y_i - x_i)$, and increased all the other coordinates of $x$ by at most $(y_i - x_i)$. We know that the directional derivative of the $i$th coordinate in the $[1, \dots, 1]^\top$ direction is negative, and that the partial derivative with respect to every $j \neq i$ is positive; this implies that the $i$th coordinate of the function at $y$ must be smaller than the $i$th coordinate of the function at $x$, so they cannot both be roots. 
\end{proof}

As a consequence of this, since the Karp--Sipser algorithm is guaranteed to be optimal in Phase 1, we know the Karp--Sipser algorithm gives a near-optimal matching whenever $\sum_j c_{ij} \bar{S_j} < e$ for all $i$. By examining the system given in \cref{lem:fixed_points_are_what_we_want}, however, we find that this is not a necessary condition for subcriticality. In fact, it is possible to achieve subcriticality even when every class has expected degree more than $e$. In Figure 1, we give a plot demonstrating this for the two-class model where one class has internal edge probability $\frac{5.6}{n}$ and external edge probability $\frac{5.6}{n}$, and the other class has no internal edge probability. Note that in this case, the expected degree of a vertex in the first class is $5.6 > 2 e$, and the expected degree of a vertex in the second class is $2.8 > e$. If we removed either the across edges or the internal edges, the model would be supercritical; however, both edge sets together cancel out and become subcritical again.

\begin{figure}[H]
\centering
\includegraphics[width=8cm]{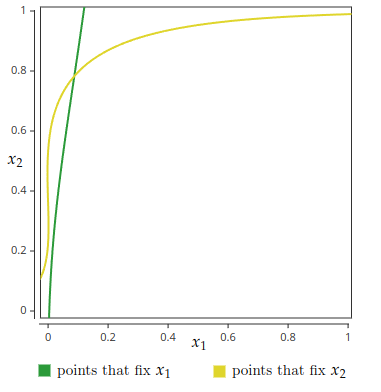}
\caption{A plot showing that for $c_{11} = 5.6$, $c_{12} = 5.6$, $c_{22} = 0$, the system has only 1 fixed point, and so we are subcritical. Here the horizontal axis is $x_1$ and the vertical axis is $x_2$; the green line is the set of points fixing $x_1$ and the yellow line is the set of points fixing $x_2$. Note that they have only a single intersection.}
\end{figure}

The values of $5.6$, $5.6$, $0$ were chosen above to make the single fixed point easily visible in the plot; it is possible to make these parameters larger. When there are 2 classes, we find that it's possible to make $c_{11}$ as high as $33.5$ while still maintaining criticality (by setting $c_{12} = 2 e+0.0514$, $c_{22} = 0$), but it is not possible to make it higher than $34$.

Qualitatively, the reason for this behaviour is that, when there are no internal edges in the second class, there is a back-and-forth thinning that can happen. That is, once we strip off all the degree-1 vertices lying entirely within the first class, the number of vertices in the first class becomes lower, so the average degree of vertices in the second class has now been decreased. Removing all the vertices of degree 1 in the second class in turn thins the first class, and it becomes possible to remove almost all the vertices in the graph, even though either of the two edge sets on their own wouldn't have been possible to remove.

\subsection{Failure of the Karp--Sipser algorithm}

In the previous sections we gave two instances where we can guarantee that the Karp--Sipser algorithm achieves a near-optimal matching, both which use essentially the same framework as the Erd\H{o}s--R\'enyi case (i.e., showing that it can achieve within $o(n)$ of a perfect matching during Phase 2, either because all degrees are close to 2 in the equitable case, or because the entire remaining graph has $o(n)$ vertices in the subcritical case). However, the algorithm does \textbf{not} return a near-optimal matching in general, and even when it does, this analytical framework does not always work. 

For example, consider a stochastic block model with 4 classes, all of size $n/4$, and the following probability matrix:
\[
\begin{bmatrix}
0 & \frac{100}{n} & 0 & 0 \\
\frac{100}{n} & 0 & \frac{10000}{n} & 0 \\
0 & \frac{10000}{n} & 0 & \frac{100}{n} \\
0 & 0 & \frac{100}{n} & 0
\end{bmatrix},
\]
see Figure~\ref{fig:example} for an illustration. 
Here, we expect the true matching size to be very close to perfect. Even if we ignore all of the edges between classes 2 and 3 entirely, we are left with two copies of $G(n/4,n/4,p)$, for which we know the matching number (as it is equitable) -- asymptotically, we can match 

However, when we analyze the performance of the Karp--Sipser algorithm, we find a different story. Phase 1 finishes very quickly, because the graph is dense enough that very few degree 1 vertices are created. Then, in Phase 2, for a long time we are in the regime of short runs as described in our analysis of the equitable case; in this regime, the algorithm chooses many of its edges uniformly at random, and so likely choose many of them from between classes 2 and 3. Every edge choosen between classes 2 and 3, however, effectively decreases the matching number by 1. Formalizing this argument reveals that the Karp--Sipser algorithm on this graph finds a matching containing only slightly more than 50\% of the vertices.

\begin{figure}[H]%
    \centering
    \subfloat[\centering Karp--Sipser]{{\includegraphics[width=5cm]{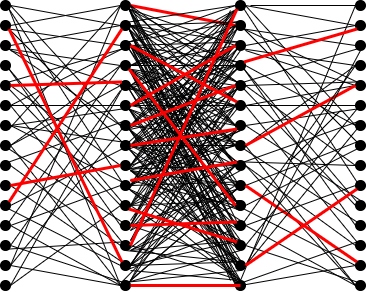} }}%
    \qquad
    \subfloat[\centering Optimal]{{\includegraphics[width=5cm]{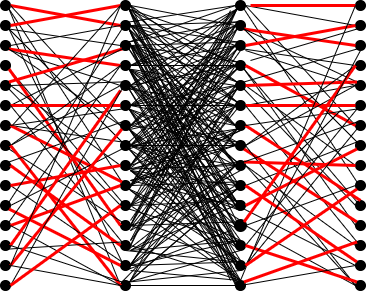} }}%
    \caption{An example on which the Karp--Sipser algorithm performs suboptimally.}%
    \label{fig:example}%
\end{figure}

This suggests the need for a modified version of the Karp--Sipser algorithm that is given the stochastic block model parameters, and takes into account the label classes of the vertices. We propose one such algorithm here:

\begin{algorithm}[H]
\caption{Label-Aware Karp--Sipser}
\begin{algorithmic}[1]
  \State $M \gets \emptyset$
  \State $V \gets V(G)$
  \While{$E(G \cap V) \neq \emptyset$}
    \If{exists $v \in V$ of degree 1 -- that is, such that exactly one $u \in V$ has $(uv) \in E(G)$}
        \State Choose such a $v$ uniformly at random over all degree 1 vertices
        \State $M \gets M \cup \{(uv)\}$
        \State $V \gets V \setminus \{u, v\}$
    \Else 
        \State $\text{Badness}(i) \gets \sum_j \theta_j \delta_{ij}$ for each class $i$
        \State $i, j \gets \min_{i, j \text{ such that } \exists u \in S_i, v \in S_j, (uv) \in E(G \cap V)} \text{Badness}(i) + \text{Badness}(j)$
        \State Choose an edge $(uv)$ uniformly randomly over all $u \in S_i \cap V$, $v \in S_j \cap V$, $(uv) \in E(G)$
        \State $M \gets M \cup \{(uv)\}$
        \State $V \gets V \setminus \{u, v\}$
    \EndIf
  \EndWhile
  \State Return $M$
\end{algorithmic}
\end{algorithm}

Among all edge types that exist in the current graph, when this algorithm has to make a random choice, it chooses from the edge type that we estimate will create the fewest degree-1 vertices on this step. This algorithm performed well in our experimental simulations for a range of block model parameters; it remains open whether asymptotically, it performs near-optimally. Part of the reason for the difficulty of analysis lies in the fact that we no longer expect a perfect matching to be possible in Phase 2 -- the following section illustrates this with a simple example.

\subsection{Bipartite Erd\H{o}s--R\'enyi case}\label{bipartite-er}
The bipartite graph $G(n,n,c/n)$ with equal part sizes and i.i.d.\ edges corresponds to an equitable stochastic block model: each vertex on the left and each vertex on the right has $c$ expected neighbours. So, with high probability, the Karp--Sipser algorithm performs optimally and returns a matching the same size as that of $G(2n, {c}/{2n})$. However, the asymmetric case $G(kn, n, c/n)$ where $k \neq 1$ is not equitable, because left vertices have $c$ expected neighbours while right vertices have $kc$ of them. It is this case that we analyze in this section. Note that, as before, the algorithm is guaranteed to be optimal for Phase 1; we need to show that it's also near optimal for Phase 2. However, now even the optimal algorithm on Phase 2 is not guaranteed to find a near-perfect matching. Consider a graph with very unequal part sizes, and very high $c$; for example, $G(n, 10n, {1000}/{n})$. This graph is dense enough that we expect at least $5n$ vertices to remain on the right after Phase 1, however we know that the matching number of this graph is at most $n$ -- so, we can't actually expect Phase 2 to find a near-perfect matching. Thus, this is an example where the typical analytical framework does not apply \cite{ks, ks-revisited}. However, the usual Karp--Sipser algorithm does behave optimally in this case, and there exists a simple proof. We sketch the argument.

\begin{lemma}\label{lem:bipartite-er-offline}
    Let $F_l$, $F_r$ and $E$ denote the number of non-isolated unmatched vertices on the left, non-isolated unmatched vertices on the right, and edges in the graph at the start of Phase 2. With high probability the Karp--Sipser algorithm finds a matching of size $\min(F_l, F_r) - o(n)$ on this graph. As a consequence, the asymptotic behavior of the matching numbers bipartite Erd\H{o}s--R\'enyi graphs can be determined from solutions to the ODEs described in \cref{phase1ode}.
\end{lemma}

\begin{proof}
Assume without loss of generality that $F_l < K F_r$ for some $K > 0$ (if $F_l = F_r + o(n)$, we know the algorithm performs optimally by our equitable case analysis). First, we argue that with probability going to 1 in $n$, over the course of Phase 2 the number of non-isolated unmatched vertices on the right is always higher than the number on the left. To show this, note that at a given step of the algorithm, the action with the highest expected loss on the right compared to the left is removing a degree-1 vertex on the right. This removes 1 vertex on the left and 1 on the right, but creates $\delta_L \theta_R$ vertices of degree 1 on the right that might later be lost. Clearly, the expected loss would be bigger if instead of reducing those vertices to degree 1, we just isolated them immediately, so the right size can overall be decreasing at a rate of at most $\delta_L \theta_R$ faster than the left side. But now, note that whenever almost all vertices are of degree 2 and there are $K$ times as many vertices on the right, $\delta_L$ can be at most $K$ times as large as $\delta_R$; since we know $\delta_R \theta_R < 1$, this means $\delta_L \theta_R < K$. Continuing an argument along these lines by bounding with an ODE would allow us to show that with probability going to 1, the right side can never get meaningfully smaller than the left side.

From that point, we simply note that $\delta_R \theta_L \leq \delta_L \theta_L < 1$, so the rate of creation of degree 1 vertices on the left is always less than 1, meaning that with high probability there are never more than $o(n)$ degree 1 vertices on the left. That in turn means that with high probability we do not isolate more than $o(n)$ degree 1 vertices on the left; so, we find a matching of size $F_l - o(n)$. 
\end{proof}

The first phase of the algorithm is optimal, and then once the second phase is reached the matching number is clearly bounded by $\min(F_l, F_r)$, so this implies that the Karp--Sipser algorithm is near-optimal on $G(kn, n, c/n)$. This analysis works because, even though we the Karp--Sipser core does not admit a perfect matching, we have an upper bound on its true matching number that we can show is nearly achieved by the algorithm. This suggests a natural approach to analyzing Label-Aware Karp--Sipser in other block models would be to try to find similar bounds on the matching number of the core. (A simple case we note as an open direction is the bipartite setting where there is only one class on the left -- i.e., parameters where the probability matrix graph is ``star shaped". On such graphs, Label-Aware Karp--Sipser simplifies to ``prefer edges to the available right class with minimum average degree".) For now, however, we move on to analyze a more restrictive type of matching algorithm.  

\section{Analysis of Online Matching Heuristics}\label{sec:online}
In this section, we consider the problem of \defn{online matching} in bipartite stochastic block models. In the situation we're concerned with, the left and right parts each have $q$ classes. We assume that vertices are assigned classes uniformly and independently at random, with classes of the right vertices known to the algorithm ahead of time, and vertices on the left (along with all the edges incident to them) arriving one-at-a-time. Uniformity of the class distribution is roughly without loss of generality; variations could be approximated by further subdividing the classes.

We propose four different heuristic algorithms for this problem, presented in increasing order of complexity. The first, called \dumb, simply chooses an available edge uniformly at random at each step. We show \dumb is optimal in the equitable case. The second, \degreedy, prefers matching to vertex classes with low expected degree. The third, \shortsighted, maximizes the probability of being able to find a match on the subsequent step. Finally, \bruteforce finds the optimal online matching strategy by dynamic programming - at each step, it prefers the class which maximizes the expected size of the matching of the remaining graph. While \bruteforce is optimal, it's runtime is $\Theta(n^{q+1})$ while the first three algorithms require only constant time for each match. Unfortunately, we show that each of the others is asymptotically sub-optimal in some instances.

We identify for each of these algorithms a family of instances on which it outperforms the previous algorithm, demonstrating in particular that only \bruteforce is optimal in general. We use the following notation: 
\begin{itemize}
    \item $R^{(t)}$: set of remaining (i.e., unmatched) right vertices at time $t$
    \item $R_{i}^{(t)}$: set of remaining right vertices of class $i$ at time $t$, $i=1,2\ldots q$ 
    \item $\text{deg}(\cdot)$: expected degree of a vertex in given class
    \item $c(\cdot)$: class of a given vertex
    \item $v^{(t)}$: the left vertex revealed at time $t$ 
    \item $A^{(t)}=\{w\in R^{(t)} | (v^{(t)},w)\in E\}$: available right vertices at time $t$
    \item $C^{(t)}=\{l\in[q]|A^{(t)}\cap R_{l}^{(t)} \neq \emptyset\}$: available right classes at time $t$
\end{itemize}

Before considering these algorithms, we recall the simpler setting of $G(n,n,c/n)$. First note that any online algorithm that chooses not to match a left vertex when it is adjacent to at least one available right vertex is sub-optimal. Since, in this setting, all right vertices are indistinguishable, any online algorithm that always matches when possible is optimal. This gives the following:

\begin{fact}[Mastin and Jaillet~\cite{er-online}] 
    In the online setting, the expected size of a matching in $G(n,n,c/n)$ produced by an optimal online algorithm is given by $$\left(1-\frac{\ln(2-e^{-c})}{c}\right)n.$$
\end{fact}

The stochastic block model setting allows for more nuance; designing optimal algorithms is nontrivial in general. 

\subsection{\dumb}
First, consider the simplest possible matching heuristic. For each left vertex that arrives, \dumb chooses uniformly at random one of the available edges adjacent to that vertex to add to the matching. While \dumb is sub-optimal in general, we can show that it returns an asymptotically optimal matching in any equitable block model.

\begin{algorithm}[H]
\caption{\dumb}
\begin{algorithmic}[1]
  \State $M \gets \emptyset$
  \For{$t$ in $\{1,2\ldots n\}$}
    \Comment{vertex $v^{(t)}$ revealed}
    \If{$|A^{(t)}|>0$}
        \State Choose an element uniformly at random from $A^{(t)}$, denoted by $u$
        \State $M \gets M \cup \{(u,v)\}$
    \EndIf
  \EndFor
  \State Return $M$
\end{algorithmic}
\end{algorithm}

\begin{lemma}
    In the equitable case, i.e., when all of the vertices have the same average degree, \dumb returns an expected matching size equal to that in the bipartite Erd\H{o}s-R\'enyi case ($G(n, n, c/n)$) with $c=\frac{1}{n}\sum_j c_{ij}|R_j| = \frac{1}{q}\sum_i c_{ij}$ .
\end{lemma}
\begin{proof}
    Let $X_j^{(t)}$ denote the number of unmatched vertices of class $R_j$ at time $t$. We have that $X_j^{(0)}=|R_j|$ and $$\mathbb{P}\left(X_{j}^{(t+1)}=x_j-1|X_{j}^{(t)}=x_j\right)=\sum_{i=1}^q\frac{1}{k}\left(1-\prod_{l=1}^q(1-p_{il})^{X_l^{(t)}}\right)\frac{p_{ij}x_j}{\sum_{l=1}^q p_{il}X_l^{(t)}},$$
where we approximate the probability of choosing a vertex of $R_j$ as $\frac{p_{ij}x_j}{\sum_{l=1}^q p_{il}x_l}$ given that there is at least one available edge. Therefore, if we let $u_j$ denote the evolution of the number of unmatched vertices of $R_j$ divided by $|R_j|n$, we have the following system of differential equations (from Wormald's Theorem) describing the evolution of unmatched vertices in each class: \begin{equation}\label{eq:ujprime}
u_j' = -\frac{1}{|R_j|}\sum_{i=1}^q\frac{1}{q}\left(1-e^{-\sum_{l=1}^q c_{il}|R_l|u_l}\right)\frac{p_{ij}|R_j|u_j}{\sum_{l=1}^q p_{il}|R_l|u_l}\qquad j\in\{1, 2,\ldots q\}.\end{equation}
Note that $u_j(0)=\frac{1}{n}$ for all $j$, i.e.,the $u_j$'s start out equal. Assume $u_j(t)=u(t)$ for some function $u$, for all $j$. Then we have
\begin{equation}\label{eq:uprime}
u'=-\sum_{i=1}^q\frac{1}{q}\left(1-e^{-\sum_{l=1}^q c_{il}|R_l|u}\right)\frac{p_{ij}}{\sum_{l=1}^qp_{il}|R_l|}=-\sum_{i=1}^q\frac{1}{q}\left(1-e^{-ncu}\right)\frac{c_{ij}}{nc}=-\frac{1}{n}(1-e^{-ncu}).\end{equation}
Recall that the evolution of the fraction of unmatched vertices $x$ in the bipartite Erd\H{o}s-R\'enyi case is the solution to by $x'=e^{-cx}-1$. Note that $u=\frac{x}{n}$ is a solution to \eqref{eq:uprime} (so $u_j=u$ is a solution to \eqref{eq:ujprime}), therefore the total number of unmatched vertices at time $t$ is $$\sum_{j=1}^n |R_j|n\cdot u(t)=n^2 u(t) = nx(t),$$ as desired. 
\end{proof}

We also note that this is tight; i.e., that no algorithm can do asymptotically better than \dumb on an equitable block model. The justification for this comes from the following fact:
\begin{lemma} \label{equitable-balance-good}
    In an equitable stochastic block model, if there are a total of $|R^{(t)}|$ unmatched vertices on the right, the probability of matching on the next step is maximized when those $|R^{(t)}|$ vertices are equally distributed among all classes (that is, each right type $r$ has $|R_r^{(t)}| = \frac{|R^{(t)}|}{q}$ unmatched vertices).
\end{lemma}
\begin{proof}
    Let $\rho_r = \frac{|R_r^{(t)}|}{|R^{(t)}|}$ denote the fraction of unmatched right vertices belonging to type $r$. The probability of there being an available edge in the next step of the algorithm is (as $n \to \infty$)
    \[\frac{1}{q} \sum_l e^{-\sum_r c_{lr} \rho_r}.\]
    By AM-GM inequality, this is at most 
    \[\left( \prod_l e^{-\sum_r c_{lr} \rho_r} \right)^{1/q} = e^{-\sum_l \frac{1}{q} \sum_r c_{lr}\rho_r} = e^{-\sum_r \frac{1}{q} \left(\sum_l c_{lr}\right)\rho_r} = e^{-c},\]
    which is precisely the value obtained by setting all $\rho_r = \frac{1}{q}$
\end{proof}

From this fact, we see that no algorithm can do asymptotically better than \dumb.

\begin{theorem}The optimal competitive ratio for any online algorithm in an equitable stochastic block model is 
\[
\frac{c-\ln(2-e^{-c})}{\left(2c - x + ce^{-x} + xce^{-x}\right) },
\]
where $x$ is the smallest solution of $x = ce^{-ce^{-x}}$. 
\end{theorem}
\begin{proof}
First, note that this value is precisely what we've shown for the competitive ratio of \dumb; this can be found by dividing the asymptotic matching size we proved in that case by the offline matching size we proved in \cref{equitable_offline}. Therefore, it is left to establish that no algorithm can do better than \dumb. We proceed by contradiction - assume there is a better algorithm, and consider the expected values of $\rho_i$ as in \cref{equitable-balance-good} in this algorithm. If the $\rho_i$'s stay within $o(1)$ of equal over all time, then the algorithm looks asymptotically identical to \dumb. If at some point they become unequal, then at that time the algorithm must fall behind \dumb and never catch up, by \cref{equitable-balance-good}.
\end{proof}

Note that this value is exactly the same as the competitive ratio lower bound conjectured by Mastin and Jaillet to be tight for Erd\H{o}s--R\'enyi graphs; as a special case, we have shown that conjecture \cite{er-online}.

Throughout the rest of the paper, we examine non-equitable block models, providing specific instances where one algorithm beats another. Our pictorial notation includes nodes, which represent classes, and edges with a number $c$, which represents an edge probability between the corresponding classes of $\frac{c}{n}$. The absence of an edge represents $0$ probability of edges between classes.

\subsection{\degreedy}

\dumb does not perform optimally in some cases. Consider the following block model with one left class and two right classes:
\begin{figure}[H]
\centering
\begin{tikzpicture}
\begin{scope}[every node/.style={circle,thick,draw}]
    \node (L0) at (0,3) {$L_0$};
    \node (R0) at (3,3) {$R_0$};
    \node (R1) at (3,0) {$R_1$};
\end{scope}

\begin{scope}[>={Stealth[red]},
              every node/.style={fill=white,circle},
              every edge/.style={draw=red,very thick}]
    \path [->] (L0) edge node {$100$} (R0);
    \path [->] (L0) edge node {$1$} (R1);
\end{scope}
\end{tikzpicture}
\caption{\degreedy beats \dumb. $|L_0|=n$, $|R_0|=|R_1|=\frac{n}{2}$}
\end{figure}
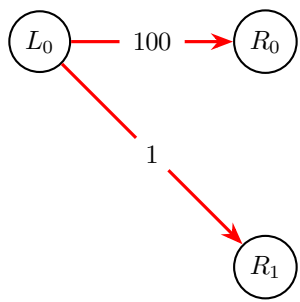

In such a scenario, it seems sensible to match $R_1$ vertices as they are available, because it is much less likely they will be available again as opposed to $R_0$ vertices (due to the average degree of $R_1$ vertices being lower than that of $R_0$ vertices). However, \dumb will naturally miss out on available $R_1$ vertices as simply choosing available edges at random heavily favors $R_0$ (the edge probability is $100$ times higher). 
This leads us to consider the following algorithm: \degreedy, which prioritizes matching right vertices with the lowest average degrees. We introduce the algorithm below -- \degreedy and all following algorithms choose which class to match to and then randomly choose an available edge from that class at each step.

At each step, \degreedy identifies the available class (i.e., a right class with at least one available edge) with the lowest average degree - if there is a tie, choose at random one of the lowest-average-degree classes. It then uniformly chooses an available edge at random from this class. 
\begin{algorithm}[H]
\caption{\degreedy}

\begin{algorithmic}[1]
  \State $M \gets \emptyset$
  \For{$t$ in $\{1,2\ldots n\}$}
    \Comment{vertex $v^{(t)}$ revealed}
    \If{$|A^{(t)}|>0$}
        \State $C=\underset{j\in C^{(t)}}{\mathrm{argmin}}$ $\text{deg}(R_j)$
        \State Choose an element uniformly at random from $C$, denoted by $c$. 
        \State Choose a vertex uniformly at random from $R_c^{(t)}$, denoted by $u$ 
        \State $M \gets M \cup \{(u,v)\}$
    \EndIf
  \EndFor
  \State Return $M$
\end{algorithmic}
\end{algorithm}

While \degreedy often outperforms \dumb, it is prone to over-correct when there is a very slight difference in average degree between right classes. Consider the following block model:
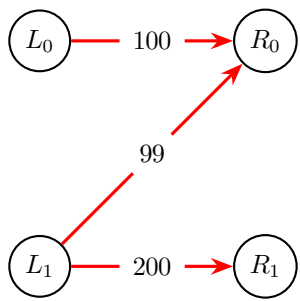
\begin{figure}[H]
\centering
\begin{tikzpicture}
\begin{scope}[every node/.style={circle,thick,draw}]
    \node (L0) at (0,3) {$L_0$};
    \node (L1) at (0,0) {$L_1$};
    \node (R0) at (3,3) {$R_0$};
    \node (R1) at (3,0) {$R_1$};
\end{scope}

\begin{scope}[>={Stealth[red]},
              every node/.style={fill=white,circle},
              every edge/.style={draw=red,very thick}]
    \path [->] (L0) edge node {$100$} (R0);
    \path [->] (L1) edge node {$99$} (R0);
    \path [->] (L1) edge node {$200$} (R1);
\end{scope}
\end{tikzpicture}
\caption{\shortsighted beats \degreedy. All classes of size $\frac{n}{2}$}
\end{figure}

In this model, \degreedy will always match to $R_0$ while it is available. Simply matching $L_0$ vertices to $R_0$ vertices and $L_1$ vertices to $R_1$ vertices gives an expected matching size of \[\frac{n}{2}\left(1-\frac{\log(2-e^{-50})}{50}+1-\frac{\log(2-e^{-100})}{100}\right)\approx .9896n.\]
However, \degreedy will match $L_1$ vertices to an $R_0$ vertex if one is available - for each occurrence of this, we are ``blocking off" an $L_0$ vertex, and reducing the size of our potential matching by $1$. The number of times an $L_1$ vertex arrives and has an available edge to $R_0$ is much greater than $.02n$, and therefore \degreedy performs worse than simply always matching within class. We would like a heuristic that somewhat prefers vertices of low expected degree, but not to such an extreme level.
%add more precise analysis here?

\subsection{\shortsighted}
The idea of \shortsighted is to, at each step, minimize the probability of being unable to match the next vertex (i.e., the probability that the next left vertex to arrive has no available edges). 
Let $u_{i}^{(t)}$ denote the number of unmatched vertices of $R_i$ at time $t$. The class(es) we choose to take an edge from at time $t$ is given by 
$$S^{(t)}=\underset{l\in C^{(t)}}{\mathrm{argmin}} \sum_{i=1}^q \frac{1}{q} \prod_{j=1}^q \left(1-p_{ij}\right)^{u_{j}^{(t)}-\mathbbm{1}\{l=j\}},$$
where the expression in the $\mathrm{argmin}$ is the probability that the next vertex has no available edges. 

\begin{algorithm}[H]
\caption{\shortsighted}

\begin{algorithmic}[1]
  \State $M \gets \emptyset$
  \State $u_{i}\gets |R_i|, i\in[q]$
  \For{$t$ in $\{1,2\ldots n\}$}
    \Comment{vertex $v^{(t)}$ revealed}
    \If{$|A^{(t)}|>0$}
        \State Compute $S^{(t)}$ per the formula above, using the current values of $u$
        \State Choose an element uniformly at random from $S^{(t)}$, denoted by $c$
        \State Choose a vertex uniformly at random from $R_c^{(t)}$, denoted by $w$ 
        \State $M \gets M \cup \{(w,v)\}$
        \State $u_c \gets u_c - 1$

    \EndIf
  \EndFor
  \State Return $M$
\end{algorithmic}
\end{algorithm}

In \cref{appendix:shortsighted}, we provide formal analysis of the behaviour of this algorithm in a broad range of settings. One might hope that looking a single step into the future is sufficient to determine the optimal class to match to, and so \shortsighted is an asymptotically optimal approach. However, this turns out not to be the case.

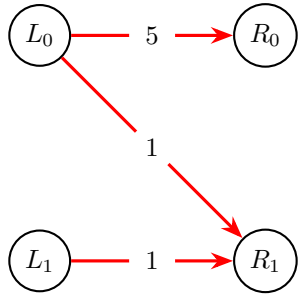
\begin{figure}[H]
\centering
\begin{tikzpicture}
\begin{scope}[every node/.style={circle,thick,draw}]
    \node (L0) at (0,3) {$L_0$};
    \node (L1) at (0,0) {$L_1$};
    \node (R0) at (3,3) {$R_0$};
    \node (R1) at (3,0) {$R_1$};
\end{scope}

\begin{scope}[>={Stealth[red]},
              every node/.style={fill=white,circle},
              every edge/.style={draw=red,very thick}]
    \path [->] (L0) edge node {$5$} (R0);
    \path [->] (L0) edge node {$1$} (R1);
    \path [->] (L1) edge node {$1$} (R1);
\end{scope}
\end{tikzpicture}
\caption{\bruteforce beats \shortsighted. All classes of size $\frac{n}{2}$}
\label{fig:kills-shortsighted}
\end{figure}

In the above model, \shortsighted prefers class $0$ when the number of unmatched vertices of class $R_0$ is at least $\frac{2\ln 2}{5}n$. Carrying out the analysis as described in \cref{appendix:shortsighted}, we obtain that the expected size of the matching is $\approx 0.574946n$. To show this is suboptimal, consider the following algorithm: prefer class $R_0$ until $.88n$ vertices have arrived, then thereafter prefer $R_1$. With essentially the same differential equation analysis as above, we find that the expected size of the matching is $\approx 0.575597n$. This algorithm out-performs \shortsighted by a linear factor in $n$ (albeit small, $\approx .0006n$), implying that \shortsighted is asymptotically sub-optimal.

%Although \shortsighted is provably worse asymptotically than \bruteforce in some cases, experimental evidence suggests that they tend to perform very similarly. We have been unable to find a model on which the their expected matching sizes differ by more than 0.1\%. This suggests that it may be possible to show that, even though \shortsighted is not optimal, it has a very similar competitive ratio to the optimal online algorithm.

%Next, we introduce \bruteforce, and subsequently prove optimality. We then compare \shortsighted and \bruteforce.

\subsection{\bruteforce}

%\bruteforce is optimal over online algorithms, as in the expected size of the matching it returns is optimal over any bipartite stochastic block model. The idea of \bruteforce is to precompute the expected size of the matching in the remaining graph at each state: $a[i_1, i_2,\ldots i_q, r, c]$, where $i_j$ denotes the number of remaining vertices of $R_j$, $r$ denotes the number of remaining left vertices to arrive, and $c$ denotes the class of the arrived vertex. Also define $av[i_1, i_2,\ldots i_k, r]=\sum_{c=1}^q\frac{1}{q}a[i_1, i_2,\ldots i_q, r, c]$. Then at each step of the algorithm, we move to the state where $av$ is maximal. Computing this array of values takes $\mathcal{O}(n^{q+1})$ time, because each entry can be computed in constant time. We present the algorithm for $2$ classes on each side ($q=2$); the generalization is clear. We compute $a$ and $av$ from the bottom up:

Note that it is not difficult to describe an optimal algorithm. \bruteforce precomputes, for every possible state of the algorithm (i.e., number of left vertices remaining and number of right vertices remaining in each class), the maximum expected matching size any strategy achieves from that state, and then always makes decisions to maximize this value. This gives the optimal expected matching size by definition. It is possible to do this precomputation in time $\mathcal{O}(n^{q+1})$ by dynamic programming.

\begin{algorithm}[H]
\caption{\bruteforce} 
\begin{algorithmic}[1]
    \State $D \gets $ empty table
    \State $D[0][i_1, \dots, i_q] \gets 0$ for all $i_1, \dots, i_q \in [n]$
    \For{$l \in [n]$}
        \For{$i_1, \dots, i_q \in [n]$}
            \State $p_j \gets \Pr\Big[\text{random left vertex matches class $j$, $\not\exists k$: $D[l-1][i_1, \dots, i_k - 1, \dots i_q] > D[l-1][i_1, \dots, i_j - 1, \dots i_q]$}\Big]$\\
            \Comment Can compute these probabilities explicitly and efficiently
            \State $D[l][i_1, \dots, i_q] \gets \Pr[\text{no matches}]\Big(D[l-1][i_1, \dots i_q]\Big) + \sum_{j \in [q]} p_j\Big(D[l-1][i_1, \dots, i_j - 1, \dots i_q] + 1\Big)$
        \EndFor
    \EndFor
    \State 
    \State $i_j \gets \text{number of right vertices in class $j$}$, $l \gets n$
    \For{$t \in [n]$, as vertices arrive}
        \Comment{vertex $v^{(t)}$ revealed}
        \If{available match}
            \State Match to available class $j$ that maximizes $D[l-1][i_1, \dots, i_j - 1, \dots i_q]$
            \State $i_j \gets i_j - 1$
        \EndIf
        \State $l \gets l - 1$
    \EndFor
\end{algorithmic}
\end{algorithm}

We conjecture based on simulations of \bruteforce that there should exist a simple, efficient algorithm that is also asymptotically optimal, however we have not found such a description. Plots of simulation of \bruteforce on the model in \cref{fig:kills-shortsighted} are included below.

\begin{figure}[H]
\centering
\includegraphics[width=8cm]{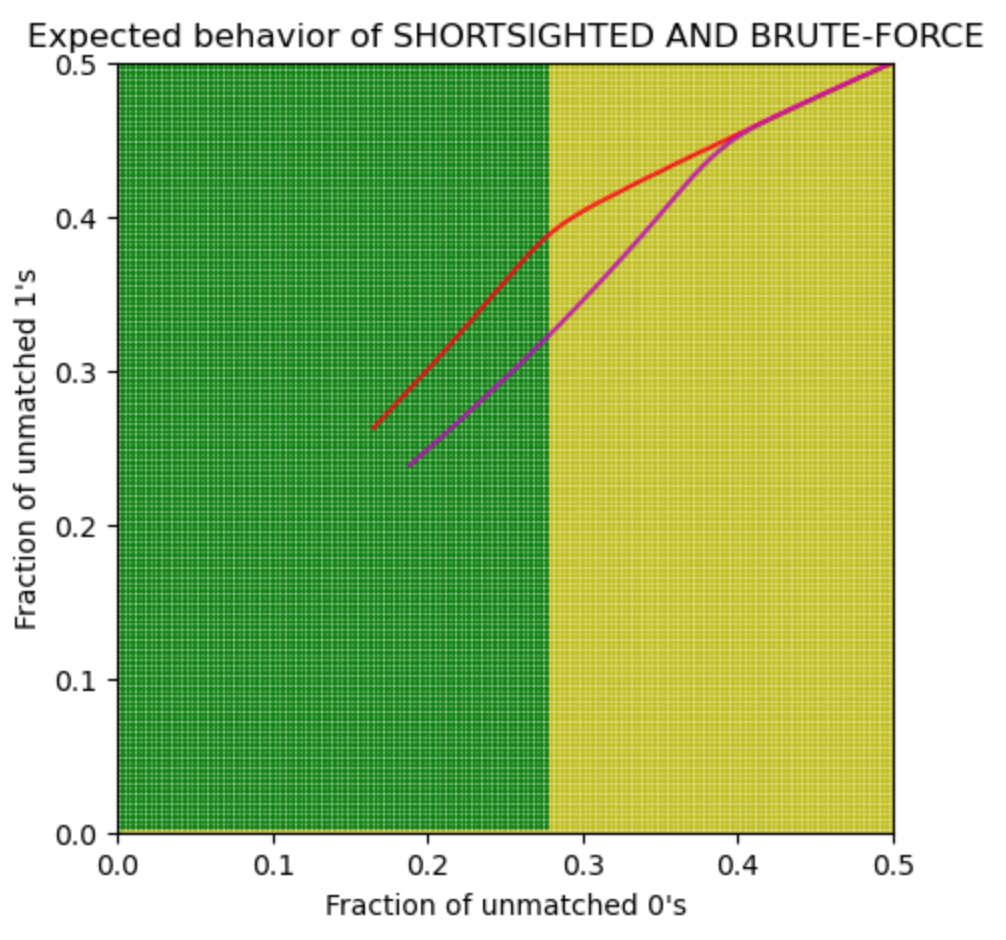}
\caption{Behavior of \shortsighted and \bruteforce for plot in \cref{fig:kills-shortsighted}, $n=300$. Yellow indicates \shortsighted's preference for class $0$, green for class $1$. The red curve is the average of the \shortsighted's path on $10,000$ instances of the model. Magenta shows \bruteforce's path.}
\end{figure}

\begin{figure}[H]
  \centering
  \includegraphics[width=.3\linewidth]{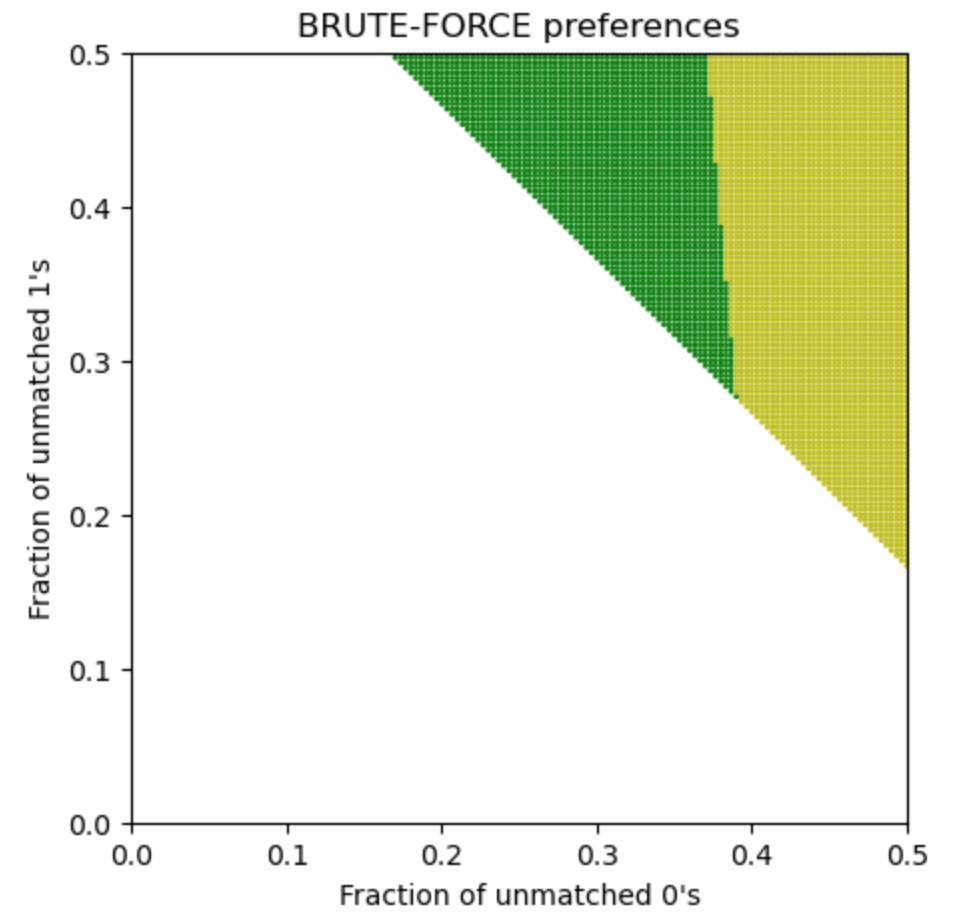}
  \includegraphics[width=.3\linewidth]{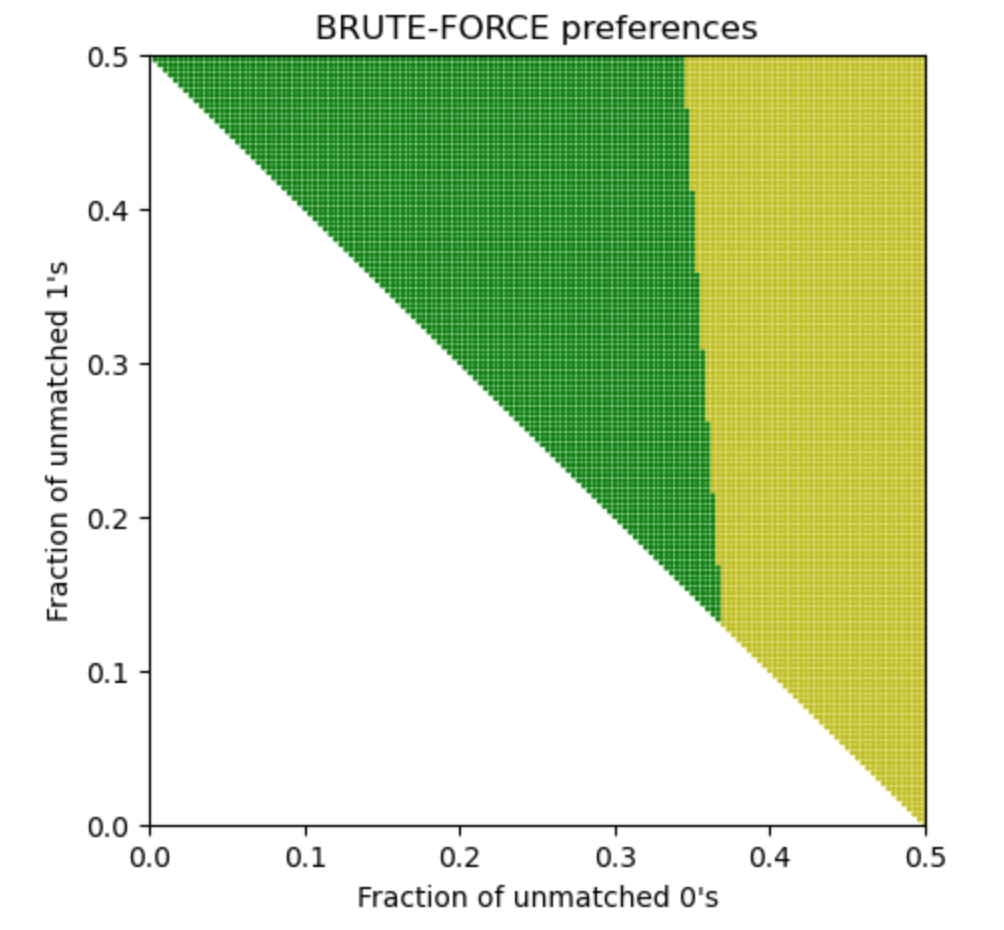}
   \includegraphics[width=.3\linewidth]{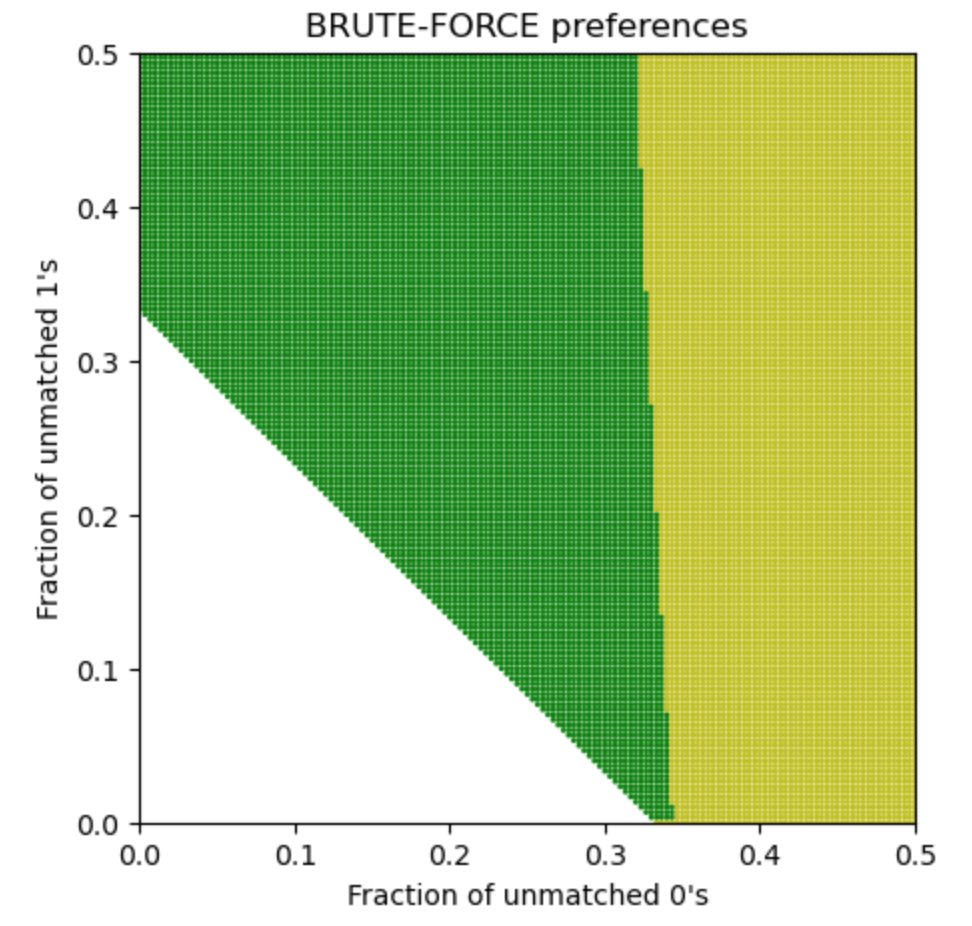}
   \includegraphics[width=.3\linewidth]{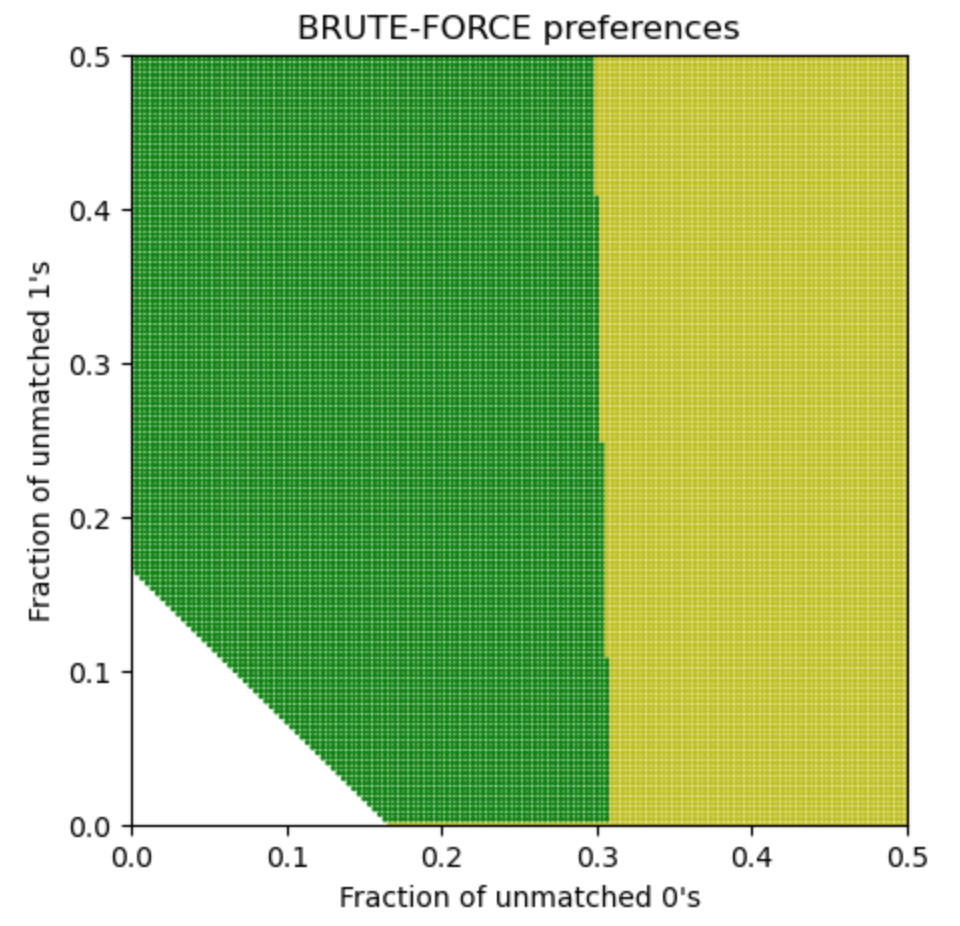}
   \includegraphics[width=.3\linewidth]{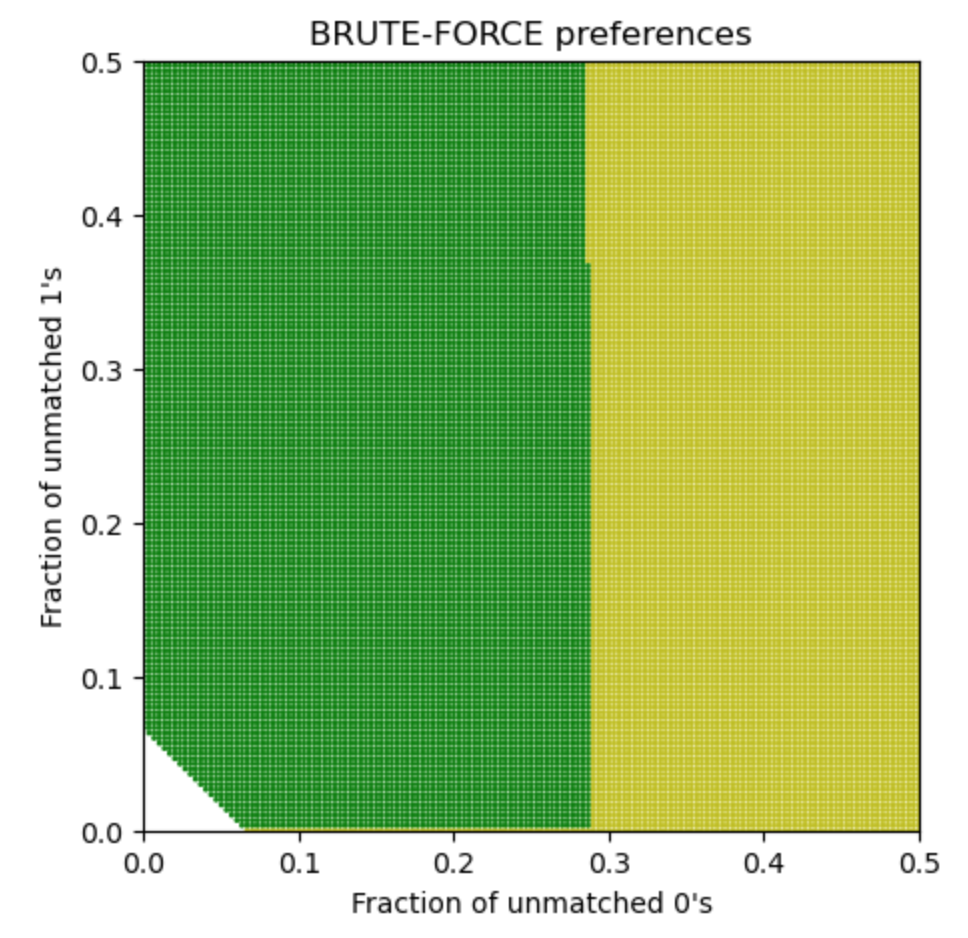}
  \includegraphics[width=.3\linewidth]{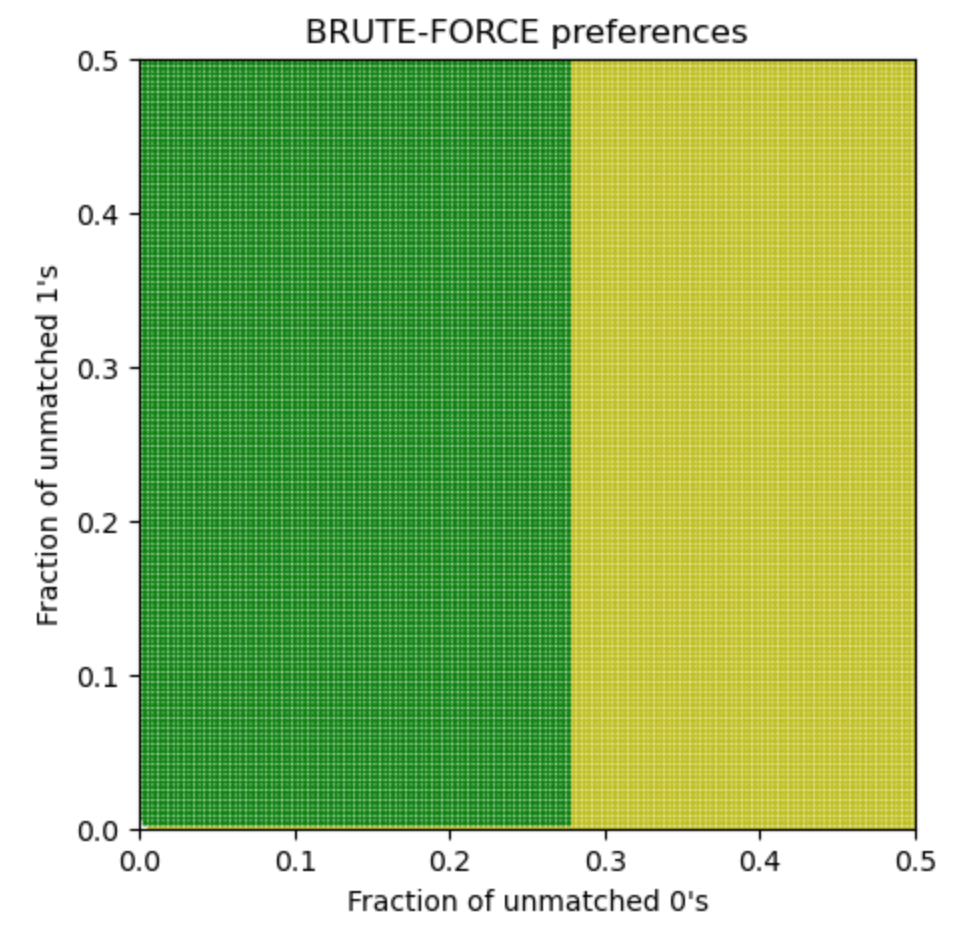}
  \caption{Preferences of \bruteforce for the model in \cref{fig:kills-shortsighted} varying over time. Yellow represents a preference for class $0$, green for class $1$. The number of vertices remaining for each plot are $200, 150, 100, 50, 20,$ and $2$ from left to right. Portions of the plot that are uncolored represent infeasible states at the given time.}
\end{figure}

\section{Conclusion}

In Erd\H{o}s--R\'enyi graphs, the structure looks uniform everywhere, and the performance of simple linear-time matching heuristics, both in the offline and online settings, is well understood. In this paper, we explored what happens when we reduce those uniformity assumptions, requiring edges to be indistinguishable on a local level but allowing large-scale global density differences (a model motivated by more realistic examples of graphs arising in nature). We showed that, so long as no class of vertices is expected to have noticeably different degrees from the others, these global density differences have no asymptotic effect on the performances of either those simple online or offline algorithms. When some vertices are allowed to be biased towards higher degree than others, however, we find that the situation becomes much more complex. There is ample ground for further work on this subject. In the offline setting, we would be interested to see classes of graphs on which the Label-Aware Karp--Sipser algorithm finds near optimal matchings even though oblivious Karp--Sipser does not. In the online setting, it would be valuable either to find a more efficient means of calculating the optimal that \bruteforce finds, or else to prove that \shortsighted (or some other reasonable alternative) always achieves a very similar competitive ratio.

\section{Acknowledgements}

This project was conducted as a part of the 2023 Summer Program for Undergraduate Research at MIT. We're grateful to all the organizers of the program for facilitating it, to Professor Elchanan Mossel for suggesting the topic, and to Professor David Jerison in particular for his helpful feedback. 

\bibliographystyle{plain}
\bibliography{citations}

\appendix
\section{Justifications of Transition Probabilities}
\subsection{Blocked Configuration Model}\label{appendix:configuration}

\begin{lemma}
Conditional on the output being simple (no multiple edges or self-loops), this corresponding blocked configuration model gives the same distribution as the original stochastic block model.
\end{lemma}
\begin{proof}    
If the output is simple, then for each pair of classes the size $m_{ij}$ of the edge set is distributed the same as the stochastic block model, and all edge sets of that size are equally likely (since there are exactly $m_{ij}!$ possible ordered lists of edges that produce that edge set, and all of those are equally likely). As in the stochastic block model, the values of $m_{ij}$ and the distribution of edges within those are independent for all pairs $i$ and $j$. So, this is the same distribution. 
\end{proof}
\begin{lemma}
If all label classes are of size $\Theta(n)$ and each pair of labels has $\mathcal{O}(n)$ edges between them, then with probability bounded away from zero the blocked configuration model produces a simple graph.
\end{lemma}
\begin{proof}
It suffices to show that with probability bounded away from zero the graph between any given pair of labels is simple, because the pairs are independent and there are only constantly many of them. Denote the set of vertices with label $i$ as $S_i$. Within a given label class $i$, the probability of being simple is~\cite{ks-revisited},
\[\left(1 - \frac{1}{|S_i|}\right)^{m_{ii}} \cdot \prod_{j=0}^{m_{ii}-1} \left(1 - \frac{j}{\binom{|S_i|}{2}}\right) 
\geq e^{-\frac{m_{ii}}{2|S_i|}-\frac{m_{ii}^2}{4|S_i|^2}} = \Theta(1).\]
Between two different label classes $i$ and $j$, the probability of being simple is 
\[\prod_{k=0}^{m_{ij}-1} \left(1 - \frac{k}{|S_i| \cdot |S_j|}\right) \geq e^{- \frac{m_{ij}}{|S_i||S_j|}} = \Theta(1).\]
So, the entire multigraph is simple with probability bounded away from 0.    
\end{proof}
Thus, any result which holds with high probability in the blocked configuration model also holds with high probability for the stochastic block model.

\subsection{Proof of Markov Property}\label{appendix:markov}

In this section we give the proof of \cref{lem:markov-holds}.

\begin{proof}
We proceed by induction. First, note that we start with a blocked configuration model with the given edge counts -- once we discard multigraphs where the number of thin and fat vertices don't fit $T_i$, $F_i$, we will still be uniform over the remaining multigraphs. Now, suppose that, conditional on the value of its tuple $Y_i$, the distribution of multigraphs after $i$ steps of the algorithm follows the blocked configuration model conditioned on number of thin and fat vertices. We would like to show that after $i+1$ steps, conditional on any fixed $Y_{i+1}$, the distribution still follows blocked configuration model. To do so, we let $Y_{i+1}$ be arbitrary, and let $G_1, G_2$ be outputs of the blocked configuration model (that is, ordered lists of edges for each pair of labels) that correspond to multigraphs with tuple $Y_{i+1}$ -- we will argue that the probability of reaching $G_1$ after $i+1$ steps is the same as the probability of reaching $G_2$, meaning that all such multigraphs are equally likely (this is what it means to follow blocked configuration model distribution). The way to show this is a straightforward counting argument: enumerate all the ways to reach $G_1$ or $G_2$ from a multigraph in the previous step, then appeal to the inductive hypothesis to argue that the equality of these counts implies equal probabilities. To lay out this argument explicitly, start with either $G_1$ or $G_2$ and imagine performing the following "undoing" process:
\begin{enumerate}
    \item First, choose a value for the tuple $Y_i$ of the multigraph preceding $G$. Of course, for some choices of $Y_i$ there are no ways to yield $G$ on the next step (for instance, if $Y_i$ has fewer total edges than $Y_{i+1}$, there is no way to choose a graph with tuple $Y_i$ and get a multigraph with tuple $Y_{i+1}$ after a step of Karp--Sipser) -- the goal is just to show that the number of ways to yield $G_1$ is always the same as the number of ways to yield $G_2$.
    \item Next, choose which two label classes $a$ and $b$ the $i+1$st step of the algorithm chose an edge between. (Again, it is possible that some of these choices won't correspond to any valid predecessor states.)
    \item When the algorithm chooses an edge for the matching, it removes all edges adjacent to the endpoints from the multigraph. So, the edge chosen on step $i + 1$ must have been between two vertices that are now degree 0 in $G$. The number of degree-0 vertices in a given class $j$ is equal to the total number of vertices in that class minus $F_j + T_j$, so the number of ways to choose one degree-0 $a$-vertex and one degree-0 $b$-vertex is the same for either $G_1$ or $G_2$. We choose one such pair $(x, y)$ to be the edge chosen in step $i+1$ of the algorithm.
    \item Once we fix which $(x, y)$ was chosen, determining the multigraph state prior to step $i+1$ just entails choosing what neighbour multisets $x$ and $y$ had. We have the following constraints on what vertices we can choose for these edges sets:
    \begin{itemize}
        \item $Y_i$ encodes the number of degree-1 vertices in the preceding graph -- if this number is greater than 0, then the edge chosen by the Karp--Sipser algorithm must involve such a thin vertex, and so either $x$ or $y$ must have no additional edges, so we now choose which one.
        \item Because we know $Y_i$ and $Y_{i+1}$, we know for each pair of classes exactly how many edges the $i+1$st step removed between that pair. So, we have to ensure that the number of edges we add of each type exactly equals these values. 
        \item Similarly, the difference between $Y_i$ and $Y_{i+1}$ tells us for each class $j$ the change in the number of thin and fat vertices in that class. A valid choice of edge set must, for each class $j$, add an edge to exactly $- \Delta F_j - \Delta T_j$ degree-0 vertices of class $j$; the number of ways to choose those vertices is $\binom{\text{initial size of class $j$} - (F_j)_i - (T_j)_i}{- \Delta F_j - \Delta T_j}$. Then, we have to add edges to exactly $ - \Delta F_j$ degree-1 vertices of class $j$; there are $\binom{(T_j)_i - \Delta F_j - \Delta T_j}{- \Delta F_j}$ ways to choose those. Any remaining edges must be distributed over vertices of degree at least 2. 
    \end{itemize}
    \item Finally, once we fix which edges to add, we can choose any indices within the ordered edge-lists to add them.
\end{enumerate}
The crucial point of this is that the number of options available as we progress from one step to the next in the procedure is always the same regardless of whether we're considering $G_1$ or $G_2$ -- the restrictions on edge sets depend only on the number of edges between each pair of classes, and the number of thin/fat vertices in each class, all of which is the same for both $G_1$ and $G_2$ because we're conditioning on them having the same tuple. So, we find that for any given predecessor tuple $Y_i$, the number of ways to choose a multigraph of tuple $Y_i$, remove a valid edge, and reach $G_1$ is exactly the same as the number of ways to choose a multigraph of tuple $Y_i$, remove a valid edge, and reach $G_2$. By the inductive hypothesis, all ordered edge-lists with the same $Y$-tuple are equally likely on step $i$. The number of valid edges to remove from an ordered edge-list just depends on how many edges there are and whether there are thin vertices, both of which are captured in the tuple. So, all ways to choose a multigraph of tuple $Y_i$ and remove an edge are equally likely to have occurred on step $i+1$ of the algorithm, meaning that the resulting multigraph is equally likely to be $G_1$ or $G_2$. 
\end{proof}

\subsection{Degree Distribution}\label{appendix:degrees}
Here we prove \cref{lem:degreedist}.
\begin{proof}First of all, we are conditioning on exactly how many degree 1 vertices there are in $G_n$, and so we know that the probability of a random half-edge being incident to one of them is exactly the number of degree-1 vertices divided by the total number of half-edges. Among the remaining vertices, note that in our model all half-edges are treated indistinguishably regardless of the label of the other half. So, the question of the degree of the vertex attached to a random half-edge, conditional on that vertex having degree at least 2, is equivalent to asking ``Suppose I throw ${\left( (\sum_l E_{il}) - T_i \right)}/ {F_i}$ balls (the total number of half-edges associated to the fat vertices) into $F_i$ bins. If I condition on all bins having at least 2 balls, how many balls are in the bin the first ball landed in?". Note that the bin the first ball ends up in is a uniform random bin independent of the other balls, so the probability that it ends up with $k$ balls in it converges to the probability that a randomly selected bin has $k - 1$ balls. We claim that this distribution converges to the above truncated Poisson distribution; a proof of this is given by Aronson, Frieze, and Pittel by showing that a sum of truncated Poissons is highly concentrated around its mean, which allows us to ignore the conditioning on all the other bins \cite{ks-revisited}. 
\end{proof}

\remark{Another point we note here is that the probability that a given vertex is attached to any self-loops or multiple edges is negligible, so long as the average degrees in each class remain bounded by a constant. In this case, we know that with high probability the maximum degree in the graph is $\mathcal{O}(\frac{\log n}{\log \log n})$. If a vertex $v$ has $\mathcal{O}(\frac{\log n}{\log \log n})$ edges to class $j$, and class $j$ has at least $n^{0.01}$ vertices (otherwise we can safely ignore all effects of class $j$), then the probability that two of those edges are to the same vertex tends to 0 in $n$. Since average degree will indeed stay bounded by a constant, this justifies not worrying about multiple edges. 
%(if we had more multiple edges, we might be concerned that we're frequently estimating that we'll remove two vertices from the graph, when actually both of them are the same vertex).
}

\section{Games on Multitype Branching Processes}\label{appendix:games}
This appendix concerns games played on multitype branching processes. The setting and results are the same as those considered by Holroyd and Martin with regards to the normal game on GW trees, but extended slightly to cover the multitype case \cite{games}. At the end of this section, we will explain the connection to our claim about criticality of the Karp--Sipser core.

\subsection{Definitions}

\begin{definition}
    A game on a (possibly infinite) game tree is played as follows: 
    Both players know the full structure of the tree. On a player's turn, they choose one child of the current node to walk down to, and the next player starts at that node. A player loses when there are no valid moves (i.e., when they start their turn on a leaf node). We denote by $\mathcal{N}_d$ the set of states from which an optimal player can force a win within $d$ turns of the game, $\mathcal{P}_d$ the set of states from which a player is guaranteed to lose within $d$ steps if their opponent plays optimally, and $\mathcal{D}_d$ to be the set of states in neither $\mathcal{N}_d$ nor $\mathcal{P}_d$ (i.e., the set of states in which either player can force the game to continue for $d$ steps if the other is unwilling to lose).
\end{definition}

We will consider games played on a random game tree drawn from a multitype branching process. Although for the purposes of the results in the paper we only need to understand the case where the offspring distributions are independent Poissons in each type, we will do our proofs here for arbitrary offspring distributions.

\begin{definition}
    Consider a multitype branching process with types $1 \dots q$, where the probability of a node of type $i$ having offspring set containing exactly $o_j$ offspring of each type $j$ is $p_i(o_1, \dots, o_q)$. The \textbf{offspring generating function} for this process is the polynomial function
    \[G: [0,1]^q \rightarrow [0,1]^q\]
\[
\begin{bmatrix}
    x_1 \\
    \dots \\
    x_q \\
\end{bmatrix}
\mapsto 
\sum_{o_1 = 0}^\infty \dots \sum_{o_q = 0}^\infty \left( \prod_{j} x_j^{o_j} \cdot 
\begin{bmatrix}
     p_1(o_1, \dots, o_q) \\
    \dots \\
     p_q(o_1, \dots, o_q) \\
\end{bmatrix}
\right)
\]
\end{definition}

For fixed offspring distributions, we are interested in the values of $N_d, P_d, D_d \in [0,1]^q$, where the $i$th entry of $N_d$ is the probability that the root of this branching process is an $\mathcal{N}_d$ node, given that it is of class $i$ (likewise $P_d$ and $D_d$ are the probabilities the root is in $\mathcal{P}_d$ or $\mathcal{D}_d$, respectively). 

\subsection{Convergence}

We can calculate the values for $N_d$, $P_d$, and $D_d$ recursively. Note the following:
\begin{itemize}
    \item A vertex is in $\mathcal{N}_d$ if and only if it has at least one child in $\mathcal{P}_{d-1}$
    \item A vertex is in $\mathcal{P}_d$ if and only if all of its children are in $\mathcal{N}_{d-1}$
    \item A vertex is in $\mathcal{D}_d$ when neither of the above conditions are true (and, in particular, when $d = 0$)
\end{itemize}

We can consider each of the root's children to be independent multitype branching processes, with roots of their given types. This yields the following recursions for these values in terms of the offspring generating function:
\begin{equation*}
    N_d  = \mathbf{1} - G(\mathbf{1} - P_{d-1}), \ P_d = G(N_{d-1}) \ \text{ and }\ D_d = \mathbf{1} - N_d - P_d.
\end{equation*}
Letting $F(x) = \mathbf{1} - G(\mathbf{1}-G(x))$, these recurrences become 
\begin{equation*}
    N_{d+2} = F(N_d) \ \text{ and } \ \mathbf{1} - P_{d+2} = F(\mathbf{1} - P_d).
\end{equation*}
Note that for any $k$, $N_{2k} = N_{2k+1}$, $P_{2k} = P_{2k+1}$, and $D_{2k} = D_{2k+1}$; this is because the player starting from the root will make all the odd moves of the game, so they can force the other player to get stuck in the first $2k + 1$ steps if and only if they can force a them to get stuck in the first $2k$ steps. So, to understand the behaviour of these values for large $d$, it suffices to show that the even values of both $N_d$ and $P_d$ converge to fixed points. That is, we want to describe the eventual behavior of $F \circ F \circ \dots \circ F (\mathbf{0})$ and $F \circ F \circ \dots \circ F(\mathbf{1})$ -- if we can show that these converge to fixed values, $N_d$ and $1 - P_d$ also converge to those values, respectively. This fact is implied by the following lemma:
\begin{lemma}
    If a function $F: [0,1]^q \rightarrow [0,1]^q$ has every output coordinate continuous and non-decreasing in every input coordinate, then there exists a minimal fixed point $x$ of $F$; that is, a fixed point of $F$ such that for all $i$, the $i$th coordinate of $x$ is minimal among all the $i$th coordinates of fixed points of $F$; iterating $F$ starting from $\mathbf{0}$ will reach this point. Similarly, there exists a maximal fixed point of $F$; iterating $F$ starting from $\mathbf{1}$ will reach that point.
\end{lemma}

\begin{proof}
This statement is directly implied by Tarski's fixed point theorem, applied on the complete lattice obtained by coordinate-wise ordering on $[0,1]^q$ (that is, $x \leq y$ if and only if $x_i \leq y_i$ for all $i$) \cite{tarski}. 
\end{proof}

Thus, we have that $\lim_{d\rightarrow \infty} N_d$ is equal to the minimum fixed point of $F$, $\lim_{d\rightarrow \infty} P_d$ is equal to $\mathbf{1}$ minus the maximum fixed point of $F$ and  $\lim_{d \rightarrow \infty} D_d$ is equal to the maximum fixed point of $F$ minus the minimum fixed point of $F$.

In particular, the probability of neither player having a winning strategy within $d$ steps tends to 0 in the limit of $d$ for all root types if and only if $F$ has exactly one fixed point on $[0,1]^q$.

\subsection{Implications for Karp--Sipser core}

The connection between this analysis of games and our analysis of the Karp--Sipser algorithm comes from the following fact:

\begin{lemma}
    If the root of a tree is in $\mathcal{N}_d$ or $\mathcal{P}_d$, there exists a sequence of removals that the Karp--Sipser algorithm can make within depth $d$ that result in the removal of the root.
\end{lemma}
\begin{proof}
This fact is shown in Karp and Sipser's original paper \cite{ks}; the proof is a straightforward induction ($\mathcal{P}$ nodes correspond to vertices that can remove their parent, $\mathcal{N}$ nodes correspond to vertices that can be removed by one of their children). 
\end{proof}

So, we find that as long as $D_d \rightarrow 0$, the root of the tree is removed with probability approaching 1; by the above analysis, this occurs whenever $F$ has only one fixed point. Plugging in the case of the process where nodes of type $i$ have independently Pois($c_{ij} \bar{S_j}$) children of type $j$ yields the result we use as \cref{lem:fixed_points_are_what_we_want} in the paper.

\section{Analysis of \shortsighted with two classes}\label{appendix:shortsighted}
We restrict our attention to the case where $q=2$, where the classes are indexed by $0$ and $1$. In this case, $S^{(t)}$ can be computed generally, facilitating the analysis of \shortsighted. We consider states of the form $(i,j,r)$, analogous to $(u_{0}^{(t)}, u_{1}^{(t)}, n-t)$. The first coordinate represents the number of unmatched vertices in $R_0$, the second is the number of unmatched vertices in $R_1$, and $r$ indicates the number of vertices left to arrive. The initial state is $(|R_0|, |R_1|, n)$.

%\anote{maybe add a sentence here introducing the following analysis / the conclusions that can be drawn from it.}

Below we compute for which states $(i,j,r)$ \shortsighted prefers class $R_0$ over $R_1$. If either $i$ or $j$ is $0$, there is no choice to be made, so we consider cases where $i,j>0$. \shortsighted prefers $R_0$ when
\[ \frac{1}{2}\left(1-p_{0,0}\right)^{i-1}\left(1-p_{0,1}\right)^j+\frac{1}{2}\left(1-p_{1,0}\right)^{i-1}\left(1-p_{1,1}\right)^j\leq \frac{1}{2}\left(1-p_{0,0}\right)^{i}\left(1-p_{0,1}\right)^{j-1}+\frac{1}{2}\left(1-p_{1,0}\right)^{i}\left(1-p_{1,1}\right)^{j-1}. \]
Rearranging, we have
\[
\left(1-p_{0,0}\right)^{i-1}\left(1-p_{0,1}\right)^{j-1}\left(p_{0,0}-p_{0,1}\right)\leq \left(p_{1,1}-p_{1,0}\right)\left(1-p_{1,0}\right)^{i-1}\left(1-p_{1,1}\right)^{j-1}. \]
Now we take the following cases:
\begin{enumerate}
    \item $p_{0,0}-p_{0,1}\leq 0$, $p_{1,1}-p_{1,0}\geq 0$: the inequality is always true, so \shortsighted prefers class $0$ where $i,j>0$,
    \item $p_{0,0}-p_{0,1}\geq 0$, $p_{1,1}-p_{1,0}\leq 0$: the inequality is always false, so \shortsighted prefers class $1$ where $i,j>0$,
    \item $p_{0,0}-p_{0,1}, p_{1,1}-p_{1,0} > 0$: the states where \shortsighted will prefer class $0$ are such that
    $$(i-1)\ln\left(\frac{1-p_{0,0}}{1-p_{1,0}}\right)+(j-1)\ln\left(\frac{1-p_{0,1}}{1-p_{1,1}}\right)\leq \ln\left(\frac{p_{1,1}-p_{1,0}}{p_{0,0}-p_{0,1}}\right),$$
    \item $p_{0,0}-p_{0,1}, p_{1,1}-p_{1,0} < 0$: the states where \shortsighted will prefer class $0$ are such that
    $$(i-1)\ln\left(\frac{1-p_{0,0}}{1-p_{1,0}}\right)+(j-1)\ln\left(\frac{1-p_{0,1}}{1-p_{1,1}}\right)\geq \ln\left(\frac{p_{1,1}-p_{1,0}}{p_{0,0}-p_{0,1}}\right).$$
\end{enumerate}
We define the boundary as the line on the $x-y$ plane that separates states where class $R_0$ is preferred and states where class $R_1$ is preferred. Taking $n\rightarrow\infty$, the boundary is given by 
$$
    y=\frac{c_{1,0}-c_{0,0}}{c_{0,1}-c_{1,1}}x+\frac{\ln\left(\frac{c_{0,0}-c_{0,1}}{c_{1,1}-c_{1,0}}\right)}{c_{0,1}-c_{1,1}}n
$$
In our analysis below, it proves useful to rescale $x$ and $y$ by a factor of $\frac{1}{n}$. Our new boundary is given by
\begin{equation}
    y=\frac{c_{1,0}-c_{0,0}}{c_{0,1}-c_{1,1}}x+\frac{\ln\left(\frac{c_{0,0}-c_{0,1}}{c_{1,1}-c_{1,0}}\right)}{c_{0,1}-c_{1,1}}
\end{equation}
Note that the condition for preferring class $R_0$ is given by a linear inequality and is independent of the number of remaining vertices. This independence in time is crucial in being able to analyze \shortsighted in some cases. We refer to $P_c$ the set of $(x,y)$ where class $c$ is preferred. 
\begin{lemma}
    If the boundary does not intersect the rectangle defined by the possible range of $x$ and $y$ ($x\in [0,|R_0|/n], y\in [0,|R_1|/n]$) or if the slope of the boundary is non-positive, then \shortsighted is of the following form for some $c$ and $T$ (where $T$ may be $n$): we prefer class $R_c$ up until there are $T$ vertices left to arrive; thereafter we prefer the opposite class, $R_{1-c}$. In these cases, we can numerically compute the expected size of matching that \shortsighted produces by the following process:
\begin{enumerate}
    \item Determine which class (i.e., $c$) is preferred at the point $(|R_0|, |R_1|)$, the starting state.
    \item If $c=0$, solve the following system of differential equations with initial conditions $(a(0), b(0))=\left(\frac{|R_0|}{n},\frac{|R_1|}{n}\right)$:
    $$a'(t)=-\frac{1}{2}\left(2-e^{-c_{0,0}a(t)}-e^{-c_{1,0}a(t)}\right)$$
    $$b'(t)=-\frac{1}{2}\left(e^{-c_{0,0}a(t)}(1-e^{-c_{0,1}b(t)})+e^{-c_{1,0}a(t)}(1-e^{-c_{1,1}b(t)})\right)$$
    Otherwise, if $c=1$, solve the system below:
    $$a'(t)=-\frac{1}{2}\left(e^{-c_{0,1}b(t)}(1-e^{-c_{0,0}a(t)})+e^{-c_{1,1}b(t)}(1-e^{-c_{1,0}a(t)})\right)$$
    $$b'(t)=-\frac{1}{2}\left(2-e^{-c_{0,1}b(t)}-e^{-c_{1,1}b(t)}\right)$$
    Let $(a(t),b(t))$ be the solution to this system. If $(a(1), b(1))\in P_c$, then the expected size of the matching is simply $n(1-a(1)-b(1))$.\\
    Otherwise, if $(a(1), b(1))\notin P_c$, solve for $T$ such that $(a(T), b(T))$ lies on the boundary. Then solve the other system of equations above (i.e., solve the system corresponding to $1-c$), with initial conditions at $t=0$ being $(a(T), b(T))$. Denote the solution $(f,g)$. The expected matching size is given by $n(1-f(1-T)-g(1-T))$.
\end{enumerate}
\end{lemma}
\begin{proof}
    If the boundary does not intersect the feasible region for $(x,y)$ (i.e., $x\in[0,|R_0|/n]$, $y\in[0,|R_1|/n]$), then \shortsighted will simply prefer one class at all times. Note that any possible "path" of the algorithm may only move down and left (as the number of unmatched vertices in each class can only decreasing over time). Therefore, if the boundary intersects the feasible region but has negative slope, any possible "path" can only intersect the boundary at most once at some time, where before this time the path will live in $P_c$ and after it will live in $P_{1-c}$ for some $c$. In both cases, the number of unmatched vertices can be represented by a Markov Chain (or two), and we can use Wormald's theorem to describe its behavior through the systems of differential equations given above. We discuss the specifics of our use of Wormald's theorem in \cref{appendix:wormwald}.
\end{proof}

The case where the slope of the boundary is positive is more tricky because we cannot guarantee that the path of the algorithm only crosses the boundary at most once (though in practice, \shortsighted seems well-behaved). Nevertheless, the above lemma covers $2/3$ of the simplex of values of $p_{i,j}$.

\section{Conditions of Wormald's Theorem}\label{appendix:wormwald}
We on several occasions in this paper claim that particular Markov processes remain close to a limiting system of differential equations. In this section, we step through for each of those instances the justification of those claims. The key tool is a theorem of Wormald, restated in its general form below. Here, $n$ indexes a family of discrete time random processes, each of which has ``history" sequence $H_n \in S_n^+$. The notation $Y_t$ is shorthand for $y(H_t)$.

\begin{theorem}[Wormald \cite{wormald}]
Let $a$ be fixed. For $1 \leq l \leq a$, let $y^{(l)}: \bigcup_n S_n^+ \rightarrow \mathbb{R}$ and $f_l : \mathbb{R}^{a+1} \rightarrow \mathbb{R}$, such that for some constant $C$ and all $l$, $|y^{(l)}(h_t)| < Cn$ for all $h_t \in S_n^+$ for all $n$. Suppose also that for some function $m = m(n)$:
\begin{enumerate}
    \item for some functions $w = w(n)$ and $\lambda = \lambda(n)$ with $\lambda^4 \log n < w < n^{2/3}/\lambda$ and $\lambda \rightarrow \infty$ as $n \rightarrow \infty$, for all $l$ and uniformly for all $t < m$,
    \[\mathbb{P}\left[ |Y^{(l)}_{t+1} - Y^{(l)}_t| > \frac{\sqrt{w}}{\lambda^2 \sqrt{\log n}} \Big| H_t\right] = o(n^{-3});\]
    \item for all $l$ and uniformly over all $t < m$, we always have
    \[\mathbb{E}(Y^{(l)}_{t+1} - Y^{(l)}_t | H_t) = f_l(t/n, Y_t^{(1)} / n, \dots, Y^{(a)}_t/n) + o(1);\]
    \item for each $l$, the function $f_l$ is continuous and satisfies a Lipschitz condition on $D$, where $D$ is some bounded connected open set containing the intersection of $\{(t, z^{(1)}, \dots, z^{(a)}): t \geq 0\}$ with some neighbourhood of $\{(0, z^{(1)}, \dots, z^{(a)}): \mathbb{P}(Y^{(l)}_0 = z^{(l)}n, 1 \leq l \leq a) \neq 0$ for some $n\}$.
\end{enumerate}
Then,
\begin{enumerate}
\item For $(0, \hat{z}^{(1)}, \dots, \hat{z}^{(a)}) \in D$, the system of differential equations
\[\frac{dz_l}{ds} = f_l(s, z_1, \dots, z_a), \hspace{40 pt} l = 1,\dots, a,\]
has a unique solution in $D$ for $z_l: \mathbb{R} \rightarrow \mathbb{R}$ passing through
\[z_l(0) = \hat{z}^{(l)}, \hspace{40 pt} 1\leq l \leq a\]
and which extends to points arbitrarily close to the boundary of $D$.
\item Almost surely,
\[Y^{l}_t = n z_l(t/n) + o(n)\]
uniformly for $0 \leq t \leq \min\{\sigma n, m\}$ and for each $l$, where $z_l(t)$ is the solution in (i) with $\hat{z}^{(l)} = Y^{(l)}_0/n$, and $\sigma = \sigma(n)$ is the supremum of those $s$ to which a solution can be extended.
\end{enumerate}
\end{theorem}

\subsection{Phase 1 of the Karp--Sipser algorithm}

The first time we make use of this differential equations method is in the analysis of the first phase of the Karp--Sipser algorithm. We will outline how to apply Wormald's theorem in this case. Here, we take $Y$ as we defined it, $f$ as the derivative we wrote down for the corresponding ODE, and note that with high probability, $y^{(l)}(h_t) < 100 (\max_{ij} c_{ij})n$ for all sufficiently large $n$ (the $F_i$ and $T_i$ components of $y^{(l)}(h_t)$ are clearly bounded by $n$; whp the edge counts are all initially within a factor of 100 of their expectations, and once they start that way they never increase).  We take the stopping time $m$ to be the first time all of the $T_i$ entries of $Y$ drop below $n^{0.01}$.  Now,
\begin{enumerate}
    \item  Take for instance $w = n^{0.5}$ and $\lambda = \log n$. The probability that any vertex in the initial graph has degree polynomial in $n$ decays exponentially in $n$, and it can be observed that the magnitude of a transition is bounded by twice the maximum degree, so we certainly have the desired condition.
    \item This convergence of expected transition size of the Markov process to $f$ is precisely what is guaranteed by the convergence of our degree estimates. Note that this convergence holds as the total number of thin vertices, fat vertices, and edges is going to infinity, so taking our stopping time to be $m$ prevents border cases once these values drop down to constant sizes.
    \item To show that $f$ is Lipschitz in the neighbourhood of solutions, it suffices to show that solutions to $z'(t) = f(t, z)$ always maintain constant average degree in each label class.  To do so, note that with high probability the initial average degrees are at most $100 (\max_{ij} c_{ij})$ in each class, and then observe from the equations that whenever the average degree in a given class $i$ is more than 100 at some time $t$, $-\sum_j \mathcal{\bar{E}}_{ij}'(t) > \mathcal{\bar{F}}_{i}'(t) + \mathcal{\bar{T}}_{i}'(t)$. 
\end{enumerate}
Note that we have only shown that the conditions of the theorem hold with high probability over initial graph configurations; clearly, this is sufficient for our desired result.

\subsection{Phase 2 of the Karp--Sipser algorithm (equitable case)}
We also use Wormald's theorem to justify our analysis of the second phase of the algorithm in the equitable case. Here, $Y_i$ is the state of the tuple after the $i$th run of the algorithm -- i.e., the $i$th time where there are no thin vertices. We define $f_l$ to be the expected change in the tuple that one of these runs would incur if the transition probability estimates from Section 2.5 held exactly and did not change throughout the run. 

In order for this process to have the desired Lipschitz properties, we will need to define a more restricted domain than the entire possible space of tuples. In particular, we will fix some constants $\gamma$ and $\epsilon$, and consider the domain of $Y$ to consist only of tuples where the average degree into each class differs by at most $\gamma$, and is at least $2 + \epsilon$. The value of $\epsilon$ is chosen in the analysis; to determine $\gamma$, we observe the following:
\begin{itemize}
    \item $\lambda_i$ is defined such that $\frac{\lambda_i (e^{\lambda_i} - 1)}{e^\lambda_i - 1 - \lambda_i}$ is equal to the average degree of class $i$. This is monotonically increasing in the average degree of class $i$; so there is some constant $\lambda > 0$ such that if the average degree in class $i$ is at least $2 + \epsilon$, then $\lambda_i > \lambda$.
    \item The function $\delta_i\theta_i = \frac{\lambda_i}{1 - e^{-\lambda_i}} \cdot \frac{\lambda_i}{e^{\lambda_i} - 1}$ is bounded above by 1 and monotonically decreasing for $\lambda_i > 0$. So, there exists some constant $\eta > 0$ such that $\frac{\lambda}{1 - e^{-\lambda}} \cdot \frac{\lambda}{e^{\lambda} - 1} < 1 - 2 \eta$.
    \item $\lambda_i, \lambda_j \mapsto \frac{\lambda_i}{1 - e^{-\lambda_i}} \cdot \frac{\lambda_j}{e^{\lambda_j} - 1}$ is uniformly continuous, so there exists some constant $\kappa$ such that $|\lambda_i - \lambda_j| < \kappa$ implies that $|\left( \frac{\lambda_i}{1 - e^{-\lambda_i}} \cdot \frac{\lambda_j}{e^{\lambda_j} - 1}\right) - \left(\frac{\lambda_i}{1 - e^{-\lambda_i}} \cdot \frac{\lambda_i}{e^{\lambda_i} - 1}\right)| < \eta$ -- which, if $\lambda_i > 2 + \epsilon$, implies $\left( \frac{\lambda_i}{1 - e^{-\lambda_i}} \cdot \frac{\lambda_j}{e^{\lambda_j} - 1}\right) < 1 - \eta$.
    \item When defined on the domain $[2 + \epsilon, \infty)$, $\lambda_i$ is uniformly continuous as a function of the average degree of $i$. So, we can define $\gamma$ such that $|\text{average degree in class $i$} - \text{average degree in class $j$}| < \gamma$ implies $|\lambda_i - \lambda_j| < \kappa$ whenever average degrees are greater than $2 + \epsilon$.
\end{itemize}

We will define the stopping time $m$ of the process to be the first time it leaves this domain. Note that, as in Phase 1, the value of $|Y|$ is bounded by $100 c_{\max} n$ with probability $1 - o(1)$, since it is initially and can only decrease. We're now ready to verify the criteria of Wormald's theorem.

\begin{enumerate}
    \item Take again $w = n^{0.5}$ and $\lambda = \log n$. We can describe a single run of the algorithm as a branching process: each degree-1 vertex created by the algorithm corresponds to a node that has children according to the number of degree-2 vertices adjacent to its neighbour; this corresponds to expected offspring number of $\delta_i\theta_j < 1 - \eta$. After each removal we are uniform over graphs with the remaining statistics; so, as long as the branching process size is $o(n)$, we can treat the treat these offspring distributions roughly independently for each node (in particular, as long as the branching process is $o(n)$ with high probability, all offspring distributions have expectation at most $1 - \frac{\eta}{2}$ regardless of the values for the other nodes). To prove that $\mathbb{P}\left[ |Y^{(l)}_{t+1} - Y^{(l)}_t| > \frac{\sqrt{w}}{\lambda^2 \sqrt{\log n}} \Big| H_t\right] = o(n^{-3})$, we therefore just need to show that the probability that a Galton--Watson tree with $\mu < 1 - \frac{\eta}{2}$ reaches size $\frac{n^{0.25}}{(\log n)^{5/2}}$ is $o\left( \frac{1}{n^3} \right)$; this follows from a standard Chernoff bound. 
    \item The fact that $\mathbb{E}(Y^{(l)}_{t+1} - Y^{(l)}_t | H_t) = f_l(t/n, Y_t^{(1)} / n, \dots, Y^{(a)}_t/n) + o(1)$ holds is because we expect a constant run duration, and we know our initial degree estimates will hold with small error terms when we remove a constant portion of the graph.
    \item Continuity of $f_l$ is clear from continuity of our degree estimates and the fact that small changes in the edge densities of the graph can't bias the branching process too heavily. Similar justification that it's Lipschitz on the given domain can be found by examining the degree estimate functions.
\end{enumerate}

Since $Y_0/n$ is with high probability $o(1)$ from having equal average degrees, and this ODE keeps equal average degrees equal, with high probability the degrees stay within $o(1)$ of equal until the stopping time.

\subsection{Analysis of \shortsighted}
Define the 2-D Markov Chain $Z_t$, where the first coordinate represented the number of unmatched $R_0$ vertices while the second coordinate represents the number of unmatched $R_1$ vertices at time $t$ during a run of \shortsighted. We have that $Z_0=(|R_0|, |R_1|)$. If we are in the regime where we prefer class $0$, we have transition probabilities as follows:

\begin{align*}\mathbb{P}(Z_{t+1}=(x-1,y)|Z_t=(x,y))= \frac{1}{2}\left(1-\left(1-p_{0,0}\right)^{x}\right)+\frac{1}{2}\left(1-\left(1-p_{1,0}\right)^{x}\right)\rightarrow \frac{1}{2}\left(1-e^{-c_{0,0}x/n}\right)+\frac{1}{2}\left(1-e^{-c_{1,0}x/n}\right)\end{align*}
\begin{align*}\mathbb{P}(Z_{t+1}=(x,y-1)|Z_t=(x,y))= \frac{1}{2}\left(1-p_{0,0}\right)^{x}\left(1-\left(1-p_{0,1}\right)^{y}\right)+\frac{1}{2}\left(1-p_{1,0}\right)^{x}\left(1-\left(1-p_{1,1}\right)^{y}\right) \\ \rightarrow\frac{1}{2}e^{-c_{0,0}x/n}\left(1-e^{-c_{0,1}y/n}\right)+\frac{1}{2}e^{-c_{1,0}x/n}\left(1-e^{-c_{1,1}y/n}\right)\end{align*}
\begin{align*}\mathbb{P}(Z_{t+1}=(x,y)|Z_t=(x,y)) &= \frac{1}{2}\left(1-p_{0,0}\right)^{x}\left(1-p_{0,1}\right)^{y}+\frac{1}{2}\left(1-p_{1,0}\right)^{x}\left(1-p_{1,1}\right)^{y}\\
&\rightarrow\frac{1}{2}e^{-c_{0,0}x/n}e^{-c_{0,1}y/n}+\frac{1}{2}e^{-c_{1,0}x/n}e^{-c_{1,1}y/n}\end{align*}
We now verify the three conditions of Wormald's:
\begin{enumerate}
    \item This simply comes from the fact that at each time step, each coordinate of $Z_t$ can change by at most $1$.
    \item We use the following fact from \cite{er-online}: for $n>0$, $c\leq n/2$, and $x\in[0,1]$, we have 
    $$\biggr\rvert e^{-cx}-\left(1-\frac{c}{n}\right)^{nx}\biggr\rvert\leq \frac{c}{ne}$$
    We can apply this term-wise to each of our probabilities, giving us the desired condition.
    \item $e^x$ is Lipschitz continuous on $[0,1]$, therefore we have the third condition as well. 
\end{enumerate}
An essentially identical proof follows for the Markov chain where we prefer class $1$. Therefore we may apply Wormald's theorem to obtain the differential equations stated in the lemma. The second system of differential equations follows as well, because the initial conditions hold almost surely (within $o(n)$). \\
We have also applied Wormald's in the analysis of \dumb in the equitable case. Due to the similar structure of transition probabilities between \dumb and \shortsighted, the verification of the conditions of Wormald's is also very similar.
\end{document}